\documentclass[letterpaper, 10 pt, journal, twoside]{IEEEtran}
\usepackage{xcolor}
\usepackage{amsmath} % add option [fleqn] if want to left-align equations
\usepackage{amssymb}
\usepackage{amsfonts}

\usepackage{amsthm}
\usepackage{graphicx,url} 
\usepackage{bm}
\usepackage{float}
\usepackage{makecell}
\let\labelindent\relax
\usepackage{enumitem}
\usepackage{subcaption}
\usepackage{wrapfig}
\usepackage{multirow}
\usepackage{diagbox}

\usepackage{cite}
\usepackage{soul}  % for using strike through
\usepackage[bookmarks=true]{hyperref}
\hypersetup{
    colorlinks=true,
    pdfborderstyle={/S/U/W 1},
    linkcolor=blue,
    filecolor=magenta,      
    urlcolor= blue,
    linkbordercolor=red,% hyperlink border will be red
}
\newcommand\rurl[1]{%
  \href{http://#1}{\nolinkurl{#1}}%
}

\usepackage{comment}

\def\black#1{{\color{black}#1}}

\newcommand{\norm}[1]{\left\lVert#1\right\rVert}

\newcommand{\linfnorm}[1]{\left\lVert#1\right\rVert_{\mathcal{L}_\infty}}
\newcommand{\linfnormtrucI}[2]{\norm{#1}_{\mathcal{L}_\infty^{#2}}}
\newcommand{\linfnormtruc}[2]{\left\lVert#1\right\rVert_{\mathcal{L}_\infty^{[0,#2]}}}
\def\lonew{${\mathcal{L}_1}$ }

\def\xdot{\dot x}

\def \xstar {x^\star}
\def \ustar {u^\star}
\def \wstar {w^\star}

\def\mbR{\mathbb{R}}
\def\mbI{\mathbb{I}}

\def\mcA{\mathcal{A}}
\def\mcB{\mathcal{B}}
\def\mcC{\mathcal{C}}
\def\mcD{\mathcal{D}}

\def\mcX{\mathcal{X}}

\def\mcU{\mathcal{U}}

\def\trieq{\triangleq}

\newtheorem{theorem}{Theorem}
\newtheorem{lemma}{Lemma}

\theoremstyle{definition}  \newtheorem{definition}{Definition}
\theoremstyle{definition} \newtheorem{assumption}{Assumption}
\theoremstyle{remark}  
\newtheorem{remark}{Remark}

\usepackage{accents}
\newcommand{\ubar}[1]{\underaccent{\bar}{#1}}

\makeatletter
\def\cl@part {\@elt {chapter}}
\makeatother

\usepackage{cleveref}

\crefname{equation}{}{} %skip "eq" or "eqs". 
\crefname{lemma}{Lemma}{Lemmas}
\crefname{theorem}{Theorem}{Theorems}
\crefname{table}{Table}{Tables}
\crefname{figure}{Fig.}{Figs.}
\crefname{remark}{Remark}{Remarks}
\crefname{assumption}{Assumption}{Assumptions}
\crefname{section}{Section}{Sections}
\crefname{definition}{Definition}{Definitions}
\crefname{algorithm}{Algorithm}{Algorithms}

\makeatletter
\renewcommand*\env@matrix[1][\arraystretch]{%
  \edef\arraystretch{#1}%
  \hskip -\arraycolsep
  \let\@ifnextchar\new@ifnextchar
  \array{*\c@MaxMatrixCols c}}
\makeatother

\def\mcx{{\mathcal{X}}}
\def\mcC{{\mathcal{C}}}
\def\mcU{{\mathcal{U}}}

\def \xstar {x^\star}
\def \ustar {u^\star}
\def \wstar {w^\star}
\def \zstar {z^\star}

\def \Linf{$\mathcal{L}_\infty$}
\def \Ltwo{$\mathcal{L}_2$}

\def \rieman {Riemannian}
\def \opt {\mathcal{O}\mathcal{P}\mathcal{T}}

\def \optrccm {$\opt_{RCCM}$} 
\def \optref {$\opt_{REF}$} 

\def \xnom{x^\star}

\def \dotxnom{\dot x^\star}

\def \cI {(\textbf{C1})}
\def \cII {(\textbf{C2})}
\def \cIII {(\textbf{C3})}

\begin{document}

\title{\LARGE \bf
Tube-Certified Trajectory Tracking for Nonlinear Systems  \\ With Robust Control Contraction Metrics\\
(Extended Version)}

\author{Pan Zhao$^1$, Arun Lakshmanan$^1$, Kasey Ackerman$^1$, Aditya Gahlawat$^1$, Marco Pavone$^2$, and Naira Hovakimyan$^1$
\thanks{This work is supported by AFOSR and NSF under the NRI grants \#1830639 and \#1830554,   and RI grant \#2133656. } %, RI grant \#2133656
\thanks{$^1$P.~Zhao, A. Lakshmanan, K. Ackerman, A. Gahlawat and N.~Hovakimyan are with the Department of Mechanical Science and Engineering, University of Illinois at Urbana-Champaign, Urbana, IL 61801, USA. Email: \texttt{\{panzhao2, lakshma2, gahlawat, kaacker2, nhovakim\}@illinois.edu}. {\it(Corresponding author: P.~Zhao.)}}
\thanks{$^2$M. Pavone is with the Department of Aeronautics and Astronautics, Stanford University, Stanford, CA 94305, USA. Email: \texttt{pavone@stanford.edu}.}
} 

\maketitle
\thispagestyle{plain}
\pagestyle{plain}
% \thispagestyle{empty}
% \pagestyle{empty}

% \red{To do lists
% \begin{itemize}
%     \item Mention computational aspects related to geodesic computation in the simulation section and other sections. 
% \end{itemize}}

\vspace{-1cm}
\begin{abstract}
This paper presents an approach towards guaranteed trajectory tracking for nonlinear control-affine systems subject to external disturbances based on robust control contraction metrics (CCM) that aims to minimize the \Linf~gain from the disturbances to nominal-actual trajectory deviations. The guarantee is in the form of invariant tubes, computed offline and valid for any nominal trajectories, in which the actual states and inputs of the system are guaranteed to stay despite disturbances. Under mild assumptions, we prove that the proposed robust CCM (RCCM) approach yields  tighter tubes than an existing approach based on CCM and input-to-state stability analysis. We show how the RCCM-based tracking controller together with tubes can be incorporated into a feedback motion planning framework to plan safe trajectories for robotic systems. {Simulation results illustrate the effectiveness of the proposed method and empirically demonstrate reduced conservatism compared to the CCM-based approach.}
\end{abstract}
%Consider a delivery drone example: an RL policy trained under mild conditions may fail to stabilize the vehicle under windy conditions. We are interested in leveraging robust adaptive control methods to improve the robustness of an RL policy, learn the model uncertainties and help the policy robustly update itself for improved performance in non-stationary environments. 

\begin{IEEEkeywords}
Planning under uncertainty, robot safety, integrated planning and control,  robust  control, nonlinear systems
\end{IEEEkeywords}
\IEEEpeerreviewmaketitle

\section{Introduction}
Motion planning for robots with nonlinear and underactuated dynamics --  with guaranteed safety in the presence of uncertainties -- remains to be a challenging problem. The uncertainties can cause the robot's actual state trajectory to significantly deviate from its nominal behavior, causing collisions, especially when a nominal input trajectory is executed in an open-loop fashion (see Fig.~\ref{fig:illustration} for an illustration). {\it Feedback motion planning} (FMP) aims to mitigate the effect of uncertainties through the use of a feedback controller that tracks a nominal (or desired) trajectory. 
%For safety (with respect to constraint satisfaction and collision avoidance), a key problem in FMP is to design the tracking controller and compute a tube or ``funnel'' around the nominal trajectory where the actual trajectory in the presence of disturbances is guaranteed to stay in. 
A common practice in FMP to ensure vehicle safety with respect to dynamic constraints and collision avoidance involves design of the tracking controller and computation of a {\it tube} or {\it funnel} about a nominal trajectory which is guaranteed to contain the actual trajectory in the presence of  uncertainties or disturbances.
%The tube or ''funnel'' is often termed as a robust control invariant (RCI) tube. 
\begin{figure}[htb]
    \centering
    \includegraphics[width=0.75\columnwidth]{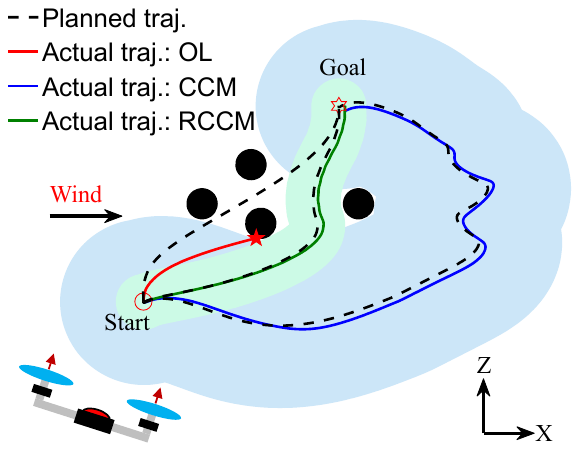}
    \caption{Planning and control of a planar VTOL vehicle in the presence of wind disturbances. Light-blue and light-green shaded areas denote the tube associated with the CCM controller from \cite{singh2019robust}, and the proposed RCCM controller, respectively. Dashed lines denote the planned trajectory without using tubes (left) and with CCM (right) and RCCM (middle) tubes. OL: open loop.}
    \label{fig:illustration}
    \vspace{-3mm}
\end{figure}
\subsection{Related Work}
Various scenarios and methods have been studied for computing and/or optimizing such tubes or funnels. For the special case of fully-actuated systems, the tubes may be computed and optimized using sliding mode control~\cite{lopez2018robust-sliding}. Approximation of such sets may also be obtained by linear reachability analysis via linearization around a nominal trajectory and treating  nonlinearities as bounded disturbances~\cite{althoff2014online-reachability}; however, these results are generally overly conservative. In~\cite{manchester2019robust-funnels}, the authors used linear analysis (i.e., propagation of ellipsoids under linearized dynamics) to compute the size of approximate invariant funnels, and further leveraged it to optimize the nominal trajectory. However, the linearity  assumption usually only holds in a small region around the nominal trajectory; furthermore, these methods usually rely on one-off offline computations and are not suitable for real-time motion planning.  
																												 
Convex programming-based verification methods such as sum of squares (SOS) programming have also gained popularity in FMP. For instance, the LQR tree algorithm in~\cite{tedrake2010lqrtree} combines local LQR feedback controllers with funnels to compose a nonlinear feedback policy to cover reachable areas. This method requires the task and environment to be pre-defined due to reliance on offline computations and is not suitable for real-time planning. The funnel library approach in~\cite{majumdar2017funnel} aims to alleviate this issue and enable online re-planning by leveraging SOS programming to compute, offline, a library of funnels around a set of nominal trajectories, in which the state is guaranteed to remain despite bounded disturbances. These funnels are then composed online for re-planning to avoid obstacles.  However, this method is still restricted to a {\it fixed set} of trajectories computed offline. 
																																			 
The concept of tube or funnel has also been explored extensively within Tube Model Predictive Control (TMPC), where one computes a tracking feedback (also termed as ancillary) controller that keeps the state within an invariant tube around the nominal MPC trajectory despite disturbances. TMPC has been extensively studied for linear systems with bounded disturbances or model uncertainties \cite{langson2004robust-tube-mpc,mayne2005robust-tube-mpc,rakovic2012parameterized-tube-mpc}. The construction of invariant tubes and ancillary controllers in the nonlinear setup is much more complicated than in the linear case. For instance, \cite{rakovic2009set} simply {assumed} the existence of a stabilizing (nonlinear) ancillary controller that results in contracting set iterates. Similarly, assuming the existence of a stabilizing feedback controller and a Lyapunov function,~\cite{marruedo2002input-tube-mpc} constructed a tube based on a Lipschitz constant of the dynamics.  This approach, although simple to apply, becomes very conservative for larger prediction horizons. In~\cite{yu2013tube-mpc}, a quadratic Lyapunov-type function with a linear auxiliary controller is computed offline, which is then used to design a robust MPC scheme for a limited class of nonlinear systems, i.e., linear systems with Lipschitz nonlinearities. For the special case of feedback linearizable systems,~\cite{lopez2019dynamic-tube-mpc} used a boundary layer sliding controller as an auxiliary controller, which enables the tube to be  parameterized as a polytope and its geometry to be co-optimized in the MPC problem. The authors of~\cite{bayer2013discrete-tube-mpc} used incremental input-to-state stability ($\delta$-ISS) for discrete-time systems to derive invariant tubes as a sublevel set of the associated $\delta$-ISS Lyapunov function, which was {\it assumed} to be given. Recently in~\cite{kohler2020computationally-tube-mpc}, for incrementally (exponentially) stabilizable nonlinear
systems subject to  nonlinear state and input dependent disturbances/uncertainty, the authors leveraged  scalar bounds of an incremental Lyapunov function, computed offline, to online predict the tube size, which is incorporated in the MPC optimization problem for constraint tightening.

Recent work has explored contraction theory within FMP. Contraction theory~\cite{lohmiller1998contraction} is a method for analyzing nonlinear systems in a differential framework and is focused on studying the convergence between pairs of state trajectories towards each other, i.e., incremental stability. It has recently been extended for constructive control design, e.g., via control contraction metrics (CCM)~for both deterministic \cite{manchester2017control} and stochastic systems \cite{tsukamoto2020neural-contraction,tsukamoto2020robust-stochastic}.  Compared to incremental Lyapunov function approaches for studying incremental stability, contraction metrics is an {\it intrinsic} characterization of incremental stability (i.e., invariant under change of coordinates); additionally, the search for a CCM and the stabilizing controller can be formulated as a convex optimization problem. Leveraging CCM, the authors of~\cite{singh2019robust} designed a feedback tracking controller for a nominal nonlinear system and derived tubes in which the actual states are guaranteed to remain despite bounded disturbances using input-to-state stability (ISS) analysis. For systems with matched uncertainties, the authors of  \cite{lakshmanan2020safe}  designed an \lonew adaptive controller to augment a nominal CCM controller and showed that the resulting tube's size could be systematically reduced by tuning some parameters of the adaptive controller, while the method in \cite{zhao2022guaranteed-contraction-imperfect} based on robust Riemannian energy conditions and disturbance estimation guaranteed exponential convergence to nominal trajectories despite the uncertainties. Finally, robust CCM was leveraged in~\cite{manchester2018rccm} to synthesize nonlinear controllers that minimize the $\mathcal{L}_2$ gain from disturbances to outputs. This method, however, does not provide tubes to quantify the {\it transient} behavior of states and inputs. 

\subsection{Contribution}
%On the theoretical side, first, 
For nonlinear control-affine systems subject to bounded disturbances, this paper presents a tracking controller based on  robust CCM (RCCM) to minimize the \Linf~gain from disturbances to state and input trajectory deviations. By solving convex optimization problems offline, the proposed RCCM scheme produces a fully nonlinear tracking controller with {\it explicit disturbance rejection} property together with {\it certificate tubes} around nominal trajectories, for both states and inputs, in which the actual state/input variables are guaranteed to stay despite disturbances. %The inherent explicit disturbance rejection property makes the resulting tracking controller different from 
In comparison, 
most of existing work in FMP and TMPC usually first designs ancillary controllers without considering the disturbances and then derives invariant tubes/funnels in the presence of disturbances using either ISS analysis (e.g., \cite{singh2019robust}), Lipschitz properties of the dynamics (e.g., \cite{marruedo2002input-tube-mpc}) or SOS verification (e.g., \cite{tedrake2010lqrtree,majumdar2017funnel}). 
We further prove, under mild assumptions, that the proposed RCCM approach yields {\it tighter} tubes than the CCM approach in \cite{singh2019robust}, which ignores the disturbance in designing the tracking controller and relies on ISS analysis to derive the tubes. As an additional contribution, we illustrate how the RCCM controller and the tubes can be incorporated into an FMP framework to plan guaranteed-safe trajectories, and {verify the proposed RCCM scheme on a planar vertical take-off and landing (VTOL) vehicle and a 3D quadrotor}. Specifically, compared to the CCM approach, our RCCM approach demonstrates improved tracking performance and reduced tube size for both states and inputs, leading to more aggressive yet safe trajectories (See Fig.~\ref{fig:illustration}). 

{\it Organization of the paper}. \cref{sec:preliminaries} states the problem and some preliminary material. \cref{sec:RCCM} presents the RCCM minimizing the \Linf ~gain and its application in designing nonlinear trajectory tracking controllers with certificate tubes for transient performance guarantee. In \cref{sec:compare-with-ccm}, the proposed RCCM controller is compared with an existing CCM controller. \cref{sec:app-mp} illustrates how the RCCM controller can be incorporated into a feedback motion planning framework. Verification of the proposed controller on simulated planar VTOL and 3D quadrotor examples is included in \cref{sec:simulations}. % Finally, \cref{sec:conclusion} concludes the paper. 

{\it Notations}.  Let $\mathbb{R}^n$, $\mathbb{R}^+$ and $\mathbb{R}^{m\times n}$ denote  the $n$-dimensional real vector space, the set of non-negative real numbers, and the set of real $m$ by $n$ matrices, respectively.
$I_n$ and ${0}$ denote an $n\times n$ identity matrix, and  a zero matrix of compatible dimensions, respectively.
$\norm{\cdot}$ denotes the $2$-norm of a vector or a matrix. 
%$\mathcal{L}_\infty$ ($\mathcal{L}_\infty^{[a,b]}$) denotes the set of signals %(whose dimensions can be deduced from the context)  
%with finite amplitude on $[0,\infty)$ (on $[a,b]$).   
The space $\mathcal{L}_{\infty e}$ is the set of signals on $[0,\infty)$ which, truncated to any finite interval $[a,b]$, have finite amplitude.
The $\mathcal{L}_\infty$- and truncated $\mathcal{L}_\infty$-norm of a function $x:\mathbb{R}^+ \rightarrow\mathbb{R}^n$ are defined as $\norm{x}_{\mathcal{L}_\infty}\triangleq \sup_{t\geq 0}\norm{x(t)}$ and $\linfnormtruc{x}{T}\triangleq \sup_{0\leq t\leq T}\norm{x(t)}$, respectively. Let $\partial_y F(x)$ denote the Lie derivative of the matrix-valued function $F$ at $x$ along the vector $y$. For symmetric matrices $P$ and $Q$, $P>Q$ ($P\geq Q$) means $P-Q$ is positive definite (semidefinite). $\langle X\rangle$ is the shorthand notation of $X +X^\top$.

\section{Problem statement and Preliminaries}\label{sec:preliminaries}
% \subsection{Problem Statement}
Consider a nonlinear control-affine system 
% \begin{equation}\label{eq:dynamics}
% \begin{split}
%     \xdot &= f(x)+B(x)u + H(x)w(t),\\
%     z &= g(x) +D(x)u + J(x)w(t)
% \end{split}
% \end{equation}
\begin{equation}\label{eq:dynamics}
\begin{split}
    \xdot &=f(x) + B(x) u + B_w(x) w,\\
    z &= g(x,u),
\end{split}
\end{equation}
where $x(t)\in\mbR^{n}$ is the state vector, $u(t)\in \mbR^m$ is the control input vector, $w(t)\in \mbR^{p}$ is the disturbance vector and $z(t)\in \mbR^q$ denotes the variables related to the performance (with  $z=x$ or $z=u$ as a special case), and $f(x)$, $B(x)$ and $B_w(x)$ are known vector/matrix functions of compatible dimensions. 
%The idea of robust CCM in  \cite{manchester2018rccm} is as follows. Assume we have a nominal trajectory $(\xstar,\ustar)$ (corresponding to $w(t)=0$), e.g., generated by a planner.
We use $b_i$ and $b_{w,i}$ to represent the $i$th column of $B(x)$ and $B_w(x)$, respectively. 

For the system in \cref{eq:dynamics}, assume we have a nominal state and input trajectory, $\xstar(\cdot)$ and $\ustar(\cdot)$, which satisfy the nominal dynamics:
\begin{equation}\label{eq:dynamics-nom}
\begin{split}
    \dot{x}^\star &=f(\xstar) + B(\xstar) \ustar + B_w(\xstar) w^\star,\\
\end{split}
\end{equation}
where $w^\star$ is the vector of nominal disturbances (including $\wstar(t)\equiv 0 $ as a special case).

For the system \cref{eq:dynamics}, this paper is focused on designing a state-feedback controller in the form of
\begin{equation}\label{eq:controller-form}
    u(t) = k(x(t),\xstar(t)) + \ustar(t)
\end{equation}
to minimize the gain from disturbance deviation, $w-\wstar$, to output deviation, $z-\zstar$, of the closed-loop system (obtained by applying the controller \eqref{eq:controller-form} to \eqref{eq:dynamics}):
% \begin{equation}\label{eq:dynamics_cl}
%  \dot x= f(x)+B(x)\left(k(x,x^\star\right)+u^\star)+H(x)w(t), \  z = g(x) +D(x)\left(k(x,x^\star)+u^\star\right)
% \end{equation}
% The differential dynamics of the closed-loop system are:
% \begin{equation}\label{eq:diff_dynamics_cl}
%   {{\dot \delta }_x} = {\mcA}{\delta _x} + \mcB{\delta _w},\quad 
%   {\delta _z} = \mcC {\delta _x} + \mcD \delta w \\ 
% \end{equation}
% where ${\mcA}\triangleq {(A + BK)}$, ${\mcC}\triangleq {(C + DK)}$, $\mcB = H$, $\mcD = J$ and $K= \frac{\partial k}{\partial x}$.
\begin{equation}\label{eq:dynamics_cl}
\begin{split}
   \dot x & = f(x) +B(x) \left( k(x,x^\star)+u^\star\right) + B_w(x)w, \\
   z & = g(x,k(x,x^\star)+u^\star)
\end{split}
\end{equation}
Formally, such gain is quantified using the concept of {\it universal} \Linf~gain  defined as follows.  Hereafter, we  use universal \Linf~gain and \Linf~gain interchangeably. 
\begin{definition} \label{defn:Linf-gain-universal}
(Universal \Linf~gain)
A control system \cref{eq:dynamics_cl} achieves a universal \Linf-gain bound of $\alpha$ if for any target trajectory $x^\star, w^\star, z^\star$ satisfying  \eqref{eq:dynamics_cl}, any initial condition $x(0)$, and input $w$ such that ${w-w^\star}
\!\in\! \mathcal{L}_{\infty e}$, the condition
 \begin{equation}\label{eq:Linf-gain-bound}
 \hspace{-2mm}  \linfnormtrucI{z\!-\!z^\star}{[0,T]}^2 \! \leq\! \alpha^2 \!\linfnormtrucI{w\!-\!w^\star}{[0,T]}^2\!+\!\beta(x(0),\xstar(0)),\ \forall T\!>\!0,
     %\max_{t\in[0,T]} \norm{z(t)-z^{\star}(t)} \leq \gamma  \max_{t\in[0,T]} \norm{w (t)}+ \beta(x(0),\xstar(0)),
    \end{equation}
 holds for a function $\beta(x_1,x_2)\!\geq\! 0$ with $\beta(x,x)\!=\!0$ for all $x$. 
\end{definition}
\begin{remark}
The \Linf-gain bound $\alpha$ in \cref{defn:Linf-gain-universal} naturally provides {\it certificate tubes} to quantify how much the actual trajectory $z(\cdot)$ %(including $x(\cdot)$ and $u(\cdot)$ as special cases) 
deviates from the nominal trajectory $\zstar(\cdot)$. %(including  $\xstar(\cdot)$ and  $\ustar(\cdot)$ as special cases). 
For instance, by setting $z=x$ and $x(0)=\xstar(0)$ and using a worst-case estimate of $\linfnormtruc{w-\wstar}{T}$, denoted by $\bar w$ (i.e., $\linfnormtruc{w-\wstar}{T}\leq \bar w$),  the inequality \cref{eq:Linf-gain-bound} implies
$x \!\in\! \Omega(\xstar)\trieq \left\{y\in \mbR^n\!: \linfnormtruc{y\!-\!\xstar}{T} \!\leq\! \alpha \bar w\right\}$ for any $T>0$. 
\end{remark}
% \ka{It's already mentioned that $z=x$ and $z=u$ are special cases. Probably don't need to repeat. See edits to above remark.}
\begin{remark}
\cref{defn:Linf-gain-universal} is inspired by the concept of universal \Ltwo~gain in \cite{manchester2018rccm}. However, unlike the \Linf~gain in \cref{defn:Linf-gain-universal}, the \Ltwo~gain does not produce tubes to quantify the {\it transient} behavior of the variable $z$.
\end{remark}
%Next, we will present robust CCM to design the controller \cref{eq:controller-form} to achieve a given \Linf-gain bound or minimize such bound. 

\subsection{Preliminaries}
%We first present some preliminaries. (Robust) 
CCM is a tool for controller synthesis to ensure incremental stability of a nonlinear system by studying the variational system, characterized by the differential dynamics. 
In this paper, we propose RCCM to design the controller \cref{eq:controller-form} to achieve or minimize an \Linf-gain bound.
The differential dynamics associated with \eqref{eq:dynamics} are given by 
\begin{equation}\label{eq:diff_dynamics}
  \begin{gathered}
  {{\dot \delta }_x} = A(x,u,w){\delta _x} + B(x){\delta _u} + B_w(x){\delta _w}, \hfill \\
  {\delta _z} = C(x,u){\delta _x} + D(x,u){\delta _u}, \hfill \\ 
\end{gathered}   
\end{equation}
where $A(x,u,w) \triangleq \frac{{\partial f}}{{\partial x}} + \sum\limits_{i = 1}^m {\frac{{\partial {b_i}}}{{\partial x}}} {u_i} + \sum\limits_{i = 1}^p {\frac{{\partial {b_{w,i}}}}{{\partial x}}} {w_i}$,  $C(x,u) \trieq \frac{{\partial g}}{{\partial x}}$ and  $D(x,u) \trieq \frac{{\partial g}}{{\partial u}}$.

Defining $K(x,\xstar) \trieq \frac{\partial k}{\partial x}$ with $k$ characterizing the control law \cref{eq:controller-form}, we obtain the differential dynamics of the closed-loop system \eqref{eq:dynamics_cl} as 
\begin{equation}\label{eq:diff_dynamics_cl}
  {{\dot \delta }_x} = {\mcA}{\delta _x} + \mcB{\delta _w},\quad 
  {\delta _z} = \mcC {\delta _x} + \mcD \delta w , 
\end{equation}
where \begin{equation}\label{eq:diff_cl-matrices}
 {\mcA}\triangleq {(A + BK)},\ \mcB = B_w,\  \mcC\triangleq {(C + DK)},\  \mcD = 0.   
\end{equation}
Our solution also involves the {\it differential \Linf ~gain}. %defined as follows.
\begin{definition} (Differential \Linf~gain)
A system with its differential dynamics represented by \eqref{eq:diff_dynamics_cl} has a differential \Linf-gain bound of $\alpha>0$ if for all $T>0$, we have
\begin{equation}\label{eq:diff-Linf-gain}
 \linfnormtrucI{\delta_z}{[0,T]}^2  \leq \alpha^2 \linfnormtrucI{\delta_w}{[0,T]}^2+\beta(x(0),\delta_x(0)),
\end{equation}
for some function $\beta(x,\delta x)$ with $\beta(x,0)=0$ for all $x$.
\end{definition}

Before proceeding to the main results, we first introduce some notations related to Riemannian geometry, most of which are from \cite{manchester2018rccm}. A Riemannian metric on $\mbR^n$ is a symmetric positive-definite matrix function $M(x)$, smooth in $x$, which defines a ``local Euclidean'' structure for any two tangent vectors $\delta_1$ and $\delta_2$ through the inner product $\langle \delta_1, \delta_2 \rangle_x\trieq \delta_1^\top M(x) \delta_2$  and the norm $\sqrt{\langle \delta_1, \delta_2 \rangle_x}$. A metric is called {\it uniformly bounded} if $a_1I\leq M(x) \leq a_2I$ holds $\forall x$ and for some scalars $a_2\geq a_1>0$.  Let $\Gamma(a,b)$ be the set of smooth paths between two points $a$ and $b$ in $\mbR^n$, where each $c\in\Gamma(a,b)$ is a piecewise smooth mapping, $c:[0,1]\rightarrow \mbR^n$,  satisfying $c(0)=a, c(1) = b$. We use the notation $c(s),~s\in[0,1]$, and $c_s(s)\trieq \frac{\partial c}{\partial s}$.
Given a metric $M(x)$, the {\it energy} of a path $c$ is defined as 
$
E(c)\trieq \int_0^1 c_s^\top M(c(s))c_s(s)ds.
$
We also use the notation $E(a,b)$ to denote the minimal energy of a path joining $a$ and $b$, i.e., $E(a,b)\trieq \inf_{c\in\Gamma(a,b)}E(c)$.

\section{Robust CCM  for Tube-Certified Trajectory Tracking}\label{sec:RCCM}
%In 
We first introduce an approach to designing a fully nonlinear controller in the form of \cref{eq:controller-form} to achieve a given \Linf-gain bound or minimize such a bound, leveraging RCCM. We then present the derivation and optimization of the certificate tubes around nominal state and control input trajectories, in which the actual states and inputs are guaranteed to stay.
% \ka{Should the Riemannian geomotry overview be in the preliminaries section?}

\subsection{RCCM for universal \Linf~gain guarantee}
Existing work, e.g., \cite{Scherer97Multi}, provides solutions to controller design for a linear time-invariant~(LTI) system for standard \Linf-gain guarantee/minimization using linear matrix inequality~(LMI) techniques. We now extend this result to nonlinear systems for differential \Linf-gain guarantee/minimization, summarized in the following lemma. 
\begin{lemma}\label{lem:linf-gain-analysis}
The closed-loop system \eqref{eq:dynamics_cl} has a differential \Linf-gain bound of $\alpha>0$ if there exists a uniformly-bounded symmetric metric $M(x)>0$ and positive  constants $\lambda$ and $\mu$ such that for all $x,w$, we have
\begin{align}
     \left[ {\begin{array}{*{20}{c}}
  {\left\langle {M\mathcal{A}} \right\rangle  + \dot M + \lambda M}&{M \mcB} \\ 
  {{\mcB^ \top }M}&{ - \mu I_p} 
\end{array}} \right] \leq 0, \hfill \label{eq:diff-Linf-gain-analysis-1} \\
  \left[ {\begin{array}{*{20}{c}}
  {\lambda M}&0&{{\mathcal{C}^ \top }} \\ 
  0&{(\alpha  - \mu )I_p}&{{\mathcal{D}^ \top }} \\ 
  \mathcal{C}&\mathcal{D}&{\alpha I_q} 
\end{array}} \right] \geq 0,  \hfill  \label{eq:diff-Linf-gain-analysis-2}
\end{align}
where $\dot M \trieq \sum_{i=1}^n \frac{\partial M}{\partial x_i}\dot x_i$ with $\dot x_i$ given by \eqref{eq:dynamics_cl}.
\end{lemma}
			   
\begin{proof}
Define $V(x,\delta_x)= \delta_x^\top M(x) \delta_x$. Applying the Schur complement to \eqref{eq:diff-Linf-gain-analysis-2} leads to 
\[\left[ {\begin{array}{*{20}{c}}
  {\lambda M}&0 \\ 
  0&{(\alpha  - \mu )I_p} 
\end{array}} \right] - {\alpha ^{ - 1}}\left[ {\begin{array}{*{20}{c}}
  {{\mathcal{C}^ \top }} \\ 
  {{\mathcal{D}^ \top }} 
\end{array}} \right]\left[ {\begin{array}{*{20}{c}}
  \mathcal{C}&\mathcal{D} 
\end{array}} \right] \geq 0.\]
Multiplying the preceding inequality by $[\delta_x^\top,\delta_w^\top]^\top$ and its transpose from the left and right, respectively, gives  $\lambda \delta _x^ \top M{\delta _x} + (\alpha  - \mu )\delta _w^ \top {\delta _w} - {\alpha ^{ - 1}}\delta _z^ \top {\delta _z} \geq0$, which implies 
\begin{equation}\label{eq:dz_V_dw}
   \delta _z^ \top {\delta _z}\leq \alpha\left(\lambda V(x,\delta _x) + (\alpha  - \mu )\delta _w^ \top {\delta _w}\right). 
\end{equation}

Multiplying  \eqref{eq:diff-Linf-gain-analysis-1} by $[\delta_x^\top,\delta_w^\top]^\top$ and its transpose from the left and right leads to $\dot V  + \lambda V- \mu \delta_w^\top \delta_w \leq 0$, which further implies $\dot V(x(t),\delta_x(t)) <0$ whenever $V(x(t),\delta_x(t)) > \frac{\mu}{\lambda} \delta_w(t)^\top \delta_w(t)$ for any $t\geq 0$. Therefore, for any $t \in[0, T]$,  $V(x(t),\delta_x(t))\leq \max\{\frac{\mu}{\lambda} \linfnormtrucI{\delta_w}{[0,T]}^2, V(x(0),\delta_x(0))\}\leq \frac{\mu}{\lambda} \linfnormtrucI{\delta_w}{[0,T]}^2 + V(x(0),\delta_x(0))$. Plugging the preceding inequality into \eqref{eq:dz_V_dw}, we obtain that for any $t \in[0, T]$, we have
\begin{equation}\label{eq:dz_V_dw2}
   \norm{\delta _z(t)}^2\leq \alpha^2 \linfnormtrucI{\delta_w}{[0,T]}^2 + \alpha \lambda V(x(0),\delta _x(0)),
\end{equation}
which is equivalent to \eqref{eq:diff-Linf-gain} with the definition of $\beta(x,\delta_x) \triangleq \alpha \lambda V(x,\delta _x)$. The proof is complete. \qed
\end{proof}
			 
\begin{remark}
In case the metric $M(x)$ depends on $x_i$, an element of $x$, whose derivative is dependent on the input $u$ (or $w$), $\dot M$ and thus the condition \eqref{eq:diff-Linf-gain-analysis-1} will depend on $u$ (or $w$). In this case, a bound on $u$ (or $w$) needs to be known in order to verify the conditions \cref{eq:diff-Linf-gain-analysis-1,eq:diff-Linf-gain-analysis-2}.
\end{remark}

%\begin{remark}
We term the metric $V(x,\delta_x) =\delta_x^\top M(x) \delta_x$ as a {\it robust CCM} (RCCM). 
Given a closed-loop system, \cref{lem:linf-gain-analysis} provides conditions to check whether a constant is a differential \Linf-gain bound of the system. We next address the problem of how to design a controller to achieve a desired universal \Linf-gain bound given an open-loop plant \cref{eq:dynamics}.
%\end{remark}

{\bf Control law construction}: Similar to \cite{manchester2018rccm}, we use $M(x)$ as a Riemannian metric to choose the path of minimum energy joining $x$ and $x^\star$ and construct the control law at any time $t$:
% \begin{equation}\label{eq:RCCM-controller}
% \begin{split}
%     \gamma (t) &= {\arg _{c \in \Gamma (x(t),{x^ \star }(t))}}\min E(c) \hfill, \\
% %   {\text{ }}u(t){\text{ }} &= {\text{ }}k(x(t),{x^ \star }(t)){\text{ }} + {\text{ }}{u^ \star }(t) \hfill \\ 
%  {\text{ }}u(t){\text{ }} &=  {\text{ }}{u^ \star }(t) + \int_0^s {K\left( {\gamma(t,s)} \right)\frac{{\partial \gamma(t,s)}}{{\partial s}}ds}, \hfill \\ 
% \end{split}  
% \end{equation}
\begin{subequations}\label{eq:RCCM-controller}
\begin{align}
     \gamma (t) & = {\arg _{c \in \Gamma (x(t),{x^ \star }(t))}}\min E(c) \hfill \label{eq:RCCM-controller-geodesic} \\
     {\text{ }}u(t){\text{ }} &=  {\text{ }}{u^ \star }(t) + \int_0^s {K\left( {\gamma(t,s)} \right)\frac{{\partial \gamma(t,s)}}{{\partial s}}ds}, \label{eq:RCCM-controller-u-signal}
\end{align}
\end{subequations}
where $
E(c)\trieq \int_0^1 c_s^\top M(c(s))c_s(s)ds
$ and the matrix function $K(\cdot )$ will be introduced later in \cref{eq:K-Y-W-relation}. Following \cite{manchester2017control}, we make the following assumption to simplify the subsequent analysis. 
\begin{assumption}\label{assump:cut-locus}
For the control system \eqref{eq:dynamics}, \eqref{eq:RCCM-controller}, the set of times $t\in[0,\infty)$ for which $x(t)$ is in the cut locus of $x^\star(t)$ has zero measure.
\end{assumption}
Without this assumption, the main results (Theorem~\ref{them:synthesis-universal-Linf-gain}) still hold if the derivative of the \rieman~energy, $E(x,\xstar)$, used in proof of Theorem~\ref{them:synthesis-universal-Linf-gain}, is replaced with its upper Dini derivative, as done in~\cite{singh2019robust}.  The main theoretical results for synthesizing a controller using RCCM to guarantee a universal \Linf-gain bound can now be presented.  
\begin{theorem}\label{them:synthesis-universal-Linf-gain}
For the plant \eqref{eq:dynamics} with differential dynamics \eqref{eq:diff_dynamics}, suppose there exists a uniformly-bounded metric $W(x)>0$,  
a matrix function $Y(x)$, and positive constants $\lambda,\ \mu$ and $\alpha$ such that
\begin{align}
   \left[ {\begin{array}{*{20}{c}}
  {\left\langle {AW + BY} \right\rangle  - \dot W + \lambda W}&B_w \label{eq:diff-Linf-synthesis-1}\\ 
  B_w^\top &{ - \mu I_p} 
\end{array}} \right] & \leq 0, \\
\left[ {\begin{array}{*{20}{c}}
  {\lambda W}& \star & \star  \\ 
  0&{(\alpha  - \mu )I_p}& \star  \\ 
  {CW + DY}&0&{\alpha I_q} 
\end{array}} \right] & \geq 0, \label{eq:diff-Linf-synthesis-2}\end{align}
for all $x,u,w$, where $\dot W \trieq \sum_{i=1}^n \frac{\partial W}{\partial x_i}\dot x_i$. Then for any target trajectory $u^\star,x^\star,w^\star$ satisfying \eqref{eq:dynamics-nom}, if Assumption~\ref{assump:cut-locus} holds, the RCCM controller \cref{eq:RCCM-controller} with 
\begin{equation}\label{eq:K-Y-W-relation}
    K(x) = Y(x)W^{-1}(x),
\end{equation}
achieves a universal $\mathcal{L}_\infty$-gain bound of $\alpha$ for the closed-loop system. 
\end{theorem}
\begin{proof}

Equations \eqref{eq:diff-Linf-synthesis-1} and \eqref{eq:diff-Linf-synthesis-2} can be transformed to \eqref{eq:diff-Linf-gain-analysis-1} and \eqref{eq:diff-Linf-gain-analysis-2}, respectively, by applying congruence transformation, defining $M(x)=W^{-1}(x)$, and leveraging \eqref{eq:K-Y-W-relation} and  \eqref{eq:diff_cl-matrices}. This is similar to the case of LMI-based state-feedback controller synthesis for LTI systems for \Linf~gain minimization \cite{Scherer97Multi}.

{At any time } $t = {t_i} \in [0,\infty ]$, consider the following smoothly parameterized paths of states, controls, disturbances, and outputs for $s \in [0,1]$:
\begin{equation}
    \begin{split}
  c(t,s) &= \gamma (t,s) \hfill \\
  v(t,s) &= {u^ \star }(t) + \int_0^s {K\left( {c(t,s)} \right)\frac{{\partial c(t,s)}}{{\partial s}}ds}  \hfill \\
  w(t,s) &= (1 - s){w^ \star }(t) + sw(t) \hfill \\
  \zeta (t,s) &= g\left( {c(t,s),v(t,s)} \right). \hfill
   \end{split}
\end{equation}
Differentiating these four paths with respect to $s$ at fixed time $t = {t_i}$ with subscript $s$ denoting $\frac{\partial}{\partial s}$ yields:
\begin{equation}\label{eq:path_diff}
    \begin{split}
  {c_s}(t,s) &= {\gamma _s}(t,s) \hfill \\
  {v_s}(t,s) &= K\left( {c(t,s)} \right){c_s}(t,s) \hfill \\
  {w_s}(t,s) &= w(t) - {w^ \star }(t) \hfill \\
  {\zeta _s}(t,s) &= C\left( {c(t,s),v(t,s)} \right){c_s}(t,s) \\
  & ~~~+ D\left( {c(t,s),v(t,s)} \right){v_s}(t,s). \hfill 
   \end{split}
\end{equation}
Now suppose that on some time interval $[{t_i},{t_i} + \epsilon )$ and for each $s \in [0,1]$, we fix the control and disturbance inputs to their values at $t = {t_i}$, and the state $c(t,s)$ evolves according to \eqref{eq:dynamics}. Here, the interval $[t_i,t_i+\epsilon]$ can be arbitrarily small to guarantee the existence of solutions.  
By changing the order of differentiation with respect to $t$ and $s$, we can show that \eqref{eq:path_diff} satisfies the closed-loop differential dynamics \eqref{eq:diff_dynamics_cl} with $\delta_x= c_s,\ \delta_z=\zeta_s, \delta_w = w_s$:
\begin{equation}\label{eq:cs_dynamics}
\begin{split}
     \dot c_s & = \mcA c_s + \mcB w_s, \\
     \zeta_s & =  \mcC c_s + \mcD w_s.
\end{split}
\end{equation}
Note that 
\begin{align}
     \frac{d}{{dt}}(c_s^ \top M{c_s}) &= c_s^ \top \dot M{c_s} + \left\langle {c_s^ \top M{{\dot c}_s}}, \right\rangle  \hfill \nonumber\\
  & = c_s^ \top \dot M{c_s} + \left\langle {c_s^ \top M\left( {\mathcal{A}{c_s} + \mcB{w_s}} \right)} \right\rangle  \hfill \label{eq:d_csMcs_1}\\
  & = c_s^ \top \left( {\dot M + \left\langle {M\mathcal{A}} \right\rangle } \right){c_s} + 2c_s^ \top M \mcB {w_s} \hfill \nonumber \\
   & \leq  - \lambda c_s^ \top M{c_s} + \mu w_s^ \top {w_s}, \hfill \label{eq:dcsMcs_lambda_mu}
\end{align}
where \eqref{eq:d_csMcs_1} is due to \eqref{eq:cs_dynamics}, and \eqref{eq:dcsMcs_lambda_mu} can be obtained by multiplying  \eqref{eq:diff-Linf-gain-analysis-1} by $[c_s^\top, w_s^\top]^\top$ and its transpose from the left and right, respectively.  Integrating \eqref{eq:dcsMcs_lambda_mu} over $s\in[0,1]$ and leveraging $w_s(t,s) = w(t) - {w^ \star }(t)$
gives $\int_0^1 {\frac{d}{{dt}}(c_s^ \top M{c_s})} ds \leq   - \lambda \int_0^1 {c_s^ \top M{c_s}} ds + \mu \int_0^1 \norm{w(t) - {w^ \star }(t)}^2 ds.$
Interchanging the differentiation and integration,  we obtain $ \frac{d}{{dt}}E(c(t))  \leq  - \lambda E(c(t)) + \mu \int_0^1 \norm{w(t) - {w^ \star }(t)}^2  ds $, i.e., 
\begin{equation} \label{eq:dE_ws_relation}
\begin{split}
   \frac{d}{{dt}}E(c(t)) & \leq   - \lambda E(c(t)) + 
  \mu \norm{w(t) - {w^ \star }(t)}^2.  \hfill 
\end{split}
\end{equation}
For sufficiently small $\epsilon$, for any $t\in[t_i,t_i+\epsilon)$, equation \eqref{eq:dE_ws_relation} indicates 
%(\blue{need to double-check the following derivations})
\begin{align}
    E(c(t)) \leq & E(c({t_i})){e^{ - \lambda (t - {t_i})}} \nonumber\\
    & + \mu \int_{t_i}^t e^{-\lambda(t-\tau)}\norm{w(\tau) - {w^ \star }(\tau)}^2 d\tau.\label{eq:Ec_t_ti}
\end{align}
Since $E(x,x^\star)$ is the minimal energy of a path joining $x$ and $x^\star$, we have $E(x(t),x^\star(t))\leq E(c(t))$ for $t\in[t_i,t_i+\epsilon)$. Furthermore, by construction, $E(c(t_i)) = E(x(t_i),x^\star(t_i))$. Therefore, \eqref{eq:Ec_t_ti} implies that $
    E(x(t),{x^ \star }(t)) \leq  E(x({t_i}),{x^ \star }({t_i})){e^{ - \lambda (t - {t_i})}} + \mu \int_{t_i}^t e^{-\lambda(t-\tau)}\norm{w(\tau) - {w^ \star }(\tau)}^2 d\tau,$
for any $t\in[t_i,t_i+\epsilon)$.
Hence, taking $\epsilon\rightarrow 0$, and, since $t_i$ was arbitrary, we have for all $t\in[0,\infty)$
\begin{equation}\label{eq:dE-E-ws}
    \frac{d}{{dt}}E(x(t),{x^ \star }(t)) \leq - \lambda E(x(t),{x^ \star }(t)) + \mu \norm{w(t) - {w^ \star }(t)}^2.
\end{equation}
Integrating the above equation from $0$ to $t$ yields
\begin{equation}\label{eq:E-x-x*-ws-relation}
  \hspace{-2mm}  E(x(t),{x^ \star }(t)) \leq E(x(0),{x^ \star }(0)){e^{ - \lambda t}} + \frac{\mu }{\lambda }\linfnormtruc{w- {w^ \star }}{t}^2.
\end{equation}
Multiplying  \eqref{eq:diff-Linf-gain-analysis-2} by $[c_s^\top,w_s^\top]^\top$ and its transpose from the left and right, respectively, gives 
$\lambda c_s^ \top M{c_s} + {w_s}(\alpha  - \mu ){w_s} - {\alpha ^{ - 1}}\zeta _s^ \top {\zeta _s} \geq0,$ which is equivalent to 
$
  \alpha^{ - 1}\zeta _s^ \top {\zeta _s} \leq\lambda c_s^ \top M{c_s} +(\alpha  - \mu ){w_s^\top}{w_s}. 
$
Integrating this inequality gives
\begin{align}
\small
    \int_0^1 {\frac{1}{\alpha}{{\left\| {{\zeta _s(t,s)}} \right\|}^2}ds}  &\leq \lambda E(c(t)) + (\alpha  - \mu ) \int_0^1 {{{\left\| {{w_s}(t,s)} \right\|}^2}ds} \nonumber \\
    & \hspace{-15mm} =  \lambda E(c(t)) + (\alpha  - \mu )\norm{w(t) - {w^ \star }(t)}^2.\label{eq:zeta_s-ws}
\end{align}
Applying the Cauchy-Schwarz inequality, we have  $\int_0^1 {{{\left\| {{\zeta _s}(t,s)} \right\|}^2}ds}  \geq{\left\| {\int_0^1 {{\zeta _s}(t,s)ds} } \right\|^2} = {\left\| {z(t) - {z^ \star(t) }} \right\|^2}$, which, together with \eqref{eq:zeta_s-ws}, leads to 
\begin{equation}\label{eq:z-z*-w-w*}
 \frac{1}{\alpha}{\left\| {z(t) - {z^ \star(t) }} \right\|^2} \leq\lambda E(c(t)) + (\alpha  - \mu )\left\| {w(t) - {w^ \star(t) }} \right\|^2,  
\end{equation}
for any $t$. Note that the preceding inequality holds for any path $c(t)$ connecting $x(t)$ and $x^\star(t)$. If we choose the path with minimal energy, i.e. $\gamma(t)$, then \eqref{eq:z-z*-w-w*} becomes
\begin{align}
& \frac{1}{\alpha}{\left\| {z(t) - {z^ \star(t)}} \right\|^2}  \nonumber\\
& \leq \lambda E(x(t),x^\star(t)) \nonumber + (\alpha  - \mu )\left\| {w(t) - {w^ \star }(t)} \right\| ^2, \nonumber \\
& \leq \lambda E(x(t),x^\star(t)) \nonumber + (\alpha  - \mu )\linfnormtruc{w - w^\star}{t}^2.\label{eq:z-z*-w-w*2}
\end{align}
Plugging \eqref{eq:E-x-x*-ws-relation} into the above inequality yields  ${\left\| {z(t)\! -\! {z^ \star }\!(t)} \right\|^2}\! \leq\! {\alpha ^2}\!\left\| {w(t)\! -\! {w^ \star\!(t) }} \right\|^2
+  \alpha \lambda E(x(0),{x^ \star }(0))e^{ - \lambda t}$ for any $t$. Therefore, for any $T>0$,
\begin{align*}
 \linfnormtrucI{z - {z^ \star }}{[0,T]}^2 \leq  {\alpha ^2}  \linfnormtrucI{w - {w^ \star }}{[0,T]}^2 +  \beta(x(0),{x^ \star }(0)),
\end{align*}
where $\beta(x,x^\star) = \alpha \lambda E(x,{x^ \star })$. The proof is complete. \qed
			
\end{proof}
%A few remarks follow. \ka{this is redundant}
\begin{remark}
From the proof of \cref{them:synthesis-universal-Linf-gain}, one can see that $W(x)$ in \cref{eq:diff-Linf-synthesis-1,eq:diff-Linf-synthesis-2} is connected with $M(x)$ in \cref{eq:diff-Linf-gain-analysis-1,eq:diff-Linf-gain-analysis-2} by $M(x) = W^{-1}(x)$. This is similar to the LTI case where a matrix equal to the inverse of a Lyapunov matrix is introduced for state-feedback control design \cite{Scherer97Multi}. We term $W(x)$ as a {\it dual} RCCM.
\end{remark}

{\bf Removal of synthesis conditions' dependence on $u$}: Condition \cref{eq:diff-Linf-synthesis-1} may depend on $u$ and $w$ due to the presence of terms $A$ and $\dot W$. Dependence on $w$ is not a significant issue as a bound on $w$ can usually be pre-established and incorporated in solving the optimization problem involving \cref{eq:diff-Linf-synthesis-1}. Since a bound on $u$ is not easy to obtain (before a controller is synthesized), the dependence of \cref{eq:diff-Linf-synthesis-1} on $u$ is undesired. To remove the dependence on $u$, we need the following condition: 
\begin{itemize}[leftmargin=10mm]
    \item[\cI] For each $i=1,\dots,m$, $\partial_{b_i} W - \left\langle \frac{\partial b_i}{\partial x}W\right\rangle =0$.
\end{itemize}
Formally, condition \cI~states that $b_i$ is a Killing vector for the metric $W$ \cite[Section III.A]{manchester2017control}.  In particular, if $B$ is  in the form of $[0, I_{m_1}]^\top$, condition \cI~requires that $W$ must not depend on the last $m_1$ state variables.
\begin{remark}
Due to the product term $\lambda W$ in \eqref{eq:diff-Linf-synthesis-1}, conditions \cref{eq:diff-Linf-synthesis-1,eq:diff-Linf-synthesis-2}  are not convex. However, since $\lambda$ is a constant, one can perform a line or bi-section search for $\lambda$. In such case, verifying the conditions  \cref{eq:diff-Linf-synthesis-1,eq:diff-Linf-synthesis-2} becomes a state-dependent LMI problem, which can be solved by gridding of the state space or using sum of square (SOS) techniques (see \cite{singh2019robust} for details).
\end{remark}

\subsection{Offline search of RCCM for \Linf~gain minimization}\label{sec:sub-opt-rccm}
The constant $\alpha$, which is an upper bound on the universal \Linf~gain, appears linearly in the condition \eqref{eq:diff-Linf-synthesis-2} of \cref{them:synthesis-universal-Linf-gain}. Therefore, one can minimize $\alpha$ when searching for $W(x)$ and $Y(x)$. To make the optimization problem feasible, one often needs to limit the states to a compact set, i.e., considering $x\in \mcX$, where $\mathcal{X}$ is a compact set. Additionally, since calculating the inverse of $W(x)$ is needed for constructing the control law due to $M(x) = W^{-1}(x)$ (detailed in \cref{sec:sub-control-law}), one may also want to enforce a lower bound, $\underline{\beta}$, on the eigenvalues of $W(x)$. Therefore, in practice, one could solve the optimization problem $\opt_{RCCM}$:
\begin{subequations}\label{eq:opt_rccm}
\begin{align}
\opt_{RCCM}: \hspace{1cm}  &\min_{W,Y,\lambda>0,\mu>0} \alpha \\
    \textup{subject to  } & \textup{Conditions } \cref{eq:diff-Linf-synthesis-1} \textup{ and } \cref{eq:diff-Linf-synthesis-2},\\
    & W(x) \geq \underline{\beta} I_n, \\
    & x\in \mcX. 
\end{align}
\end{subequations}
Note that $\opt_{RCCM}$ just needs to be solved once offline. 

\subsection{Offline optimization for refining state and input tubes}\label{sec:sub-opt_refine_u_x_tubes}
In formulating the optimization problem \optrccm~to search for $W(x)$ and $Y(x)$, the $z$ vector often contains weighted states and inputs to balance the tracking performance and control efforts. For instance, we could have $z = [(Qx)^\top,  (Ru)^\top]^\top$, where $Q$ and $R$ are some weighting matrices. After obtaining $W(x)$ and $Y(x)$, one can always derive {\it refined} \Linf-gain bounds for some {\it specific} state and input variables, $\hat z\in\mbR^l$, by re-deriving the $C$ and $D$ matrices in \eqref{eq:diff_dynamics} for $\hat z = \hat g(x,u) $, and then solving the optimization problem $\opt_{REF}$: 
\begin{subequations}\label{eq:opt_ref}
\begin{align}
\opt_{REF}: \hspace{1cm}  &\min_{\lambda>0,\mu>0} \alpha \\
    \textup{subject to  } & \textup{Conditions } \cref{eq:diff-Linf-synthesis-1} \textup{ and } \cref{eq:diff-Linf-synthesis-2},\\
    & x\in \mcX. 
\end{align}
\end{subequations}
For instance, by solving \optref, we get an \Linf-gain bound for the deviation of some states (i.e., $\linfnorm{x_\mbI-\xstar_\mbI}$, where $\mbI$ is the index set) with $\hat z = x_\mbI$, and a \Linf-gain bound for the deviation of all inputs (i.e., $\linfnorm{u-\ustar}$) with $\hat z = u$. With an \Linf-gain bound $\alpha$ (from solving \optref) and a bound $\bar w$ on the disturbances, i.e., $\linfnorm{w-\wstar} \leq \bar w$, the actual variable $\hat z$ is guaranteed to stay in a tube around the nominal variable $\hat z^\star $, i.e.,
\begin{equation}\label{eq:tube_zhat}
    \hat z \in \Omega(\hat z^\star)\trieq \{ y \in \mbR^l: \norm{y-\hat z^\star}\leq \alpha \bar w \}.
\end{equation}
Following this idea, we can easily get the tube for all or part of the states or inputs.  
\begin{remark}
The tubes obtained through \eqref{eq:tube_zhat} hold for {\it any} trajectories that satisfy the nominal dynamics \eqref{eq:dynamics-nom}, and are particularly suitable to be incorporated into {\it online} planning and predictive control schemes, e.g., tube MPC.   
\end{remark}
\subsection{Online computation of the control law}\label{sec:sub-control-law}
{\bf Geodesic computation}: Similar to other CCM or RCCM based control \cite{manchester2017control,manchester2018rccm,singh2019robust}, 
the most computationally expensive part of the proposed control law \cref{eq:RCCM-controller} lies in online computation of the geodesic $\gamma(t)$ according to \eqref{eq:RCCM-controller-geodesic} at each time instant $t$, which necessitates solving a nonlinear programming (NLP) problem. However, since the NLP problem does not involve dynamic constraints, it is much easier to solve than a nonlinear MPC problem. Following \cite{leung2017pseudospectral-geodesic}, such a problem can be efficiently solved by applying a   pseudospectral method, i.e., by discretizing the interval $[0,1]$ using the Chebyshev-Gauss-Lobatto nodes and using Chebyshev interpolating polynomials up to degree $D$ to approximate the solution. The integral in \cref{eq:RCCM-controller-geodesic} is approximated using the Clenshaw-Curtis quadrature scheme with $N>D$ nodes.

{\bf Control signal computation}: Given the solution to the geodesic problem \eqref{eq:RCCM-controller-geodesic} parameterized by a set of values $\{\gamma(s_k)\}_{k=0}^N$ and $\{\gamma_s(s_k)\}_{k=0}^N$, $s_k\in[0,1]$, the control signal can be computed according to \eqref{eq:RCCM-controller-u-signal} with the integral again approximated by the  the Clenshaw-Curtis quadrature scheme.

The control law in \eqref{eq:RCCM-controller-u-signal} is just one way to construct a control signal achieving the universal \Linf-gain bound, but it is not the only one and others may be preferable. We now show how to construct a set of robustly stabilizing controls, following \cite{manchester2018rccm}.

From the formula for first variation of energy \cite{do2013riemannian}, we have  for the derivative of energy functional at any point $x$ that is not on the cut locus of $x^\star$:
\begin{align}
  \frac{1}{2}\frac{d}{dt} E(x,x^\star) = & \gamma_s^\top(1)M(x)(f(x)+B(x)u+B_w(x)w) \nonumber \\
  & -\gamma_s^\top(0)M(\xstar)\dot{x}^\star.   \label{eq:Edot-u-relation}
\end{align}
From the proof of \cref{them:synthesis-universal-Linf-gain}, one can see that the control law in \eqref{eq:RCCM-controller} essentially tries to ensure  \eqref{eq:dE-E-ws}. Obviously, the set of control inputs for which \eqref{eq:dE-E-ws} holds is non-empty, i.e., we have 
\begin{align}
  \min_u &\max_w  \left\{ 2\gamma_s^\top(1)M(f+Bu+B_w w)-2\gamma_s^\top(0)M(\xstar)\dot{x}^\star \right. \hfill \nonumber  \\ 
& + \lambda E(x(t),{x^ \star }(t)) 
    - \mu \norm{w(t) - {w^ \star }(t)}^2\Large{\}} \leq 0,
\end{align}
where the dependence of $M$, $f$, $B$ and $B_w$ on $x$ has been omitted. The worst-case $w$ for such case is independent of $u$: $\hat w = \wstar + \frac{1}{\mu} B_w^\top(x) M(x)\gamma_s(1)$. So, for each state $x$ we could construct a set of control inputs:
\begin{align}
\hspace{-1mm}    \mathcal{U}&\trieq  \left\{u\in \mbR^m: 2\gamma_s^\top(1)M(f+Bu+B_w \hat w)-2\gamma_s^\top(0)\right. \nonumber  \\ 
& \cdot M(\xstar)\dot{x}^\star  + \lambda E(x(t),{x^\star }(t)) 
    - \frac{1}{\mu}\norm{ B_w^\top M\gamma_s(1)}^2\Large{\}} \leq 0. \nonumber
\end{align}
It should be noted that when using the preceding equation to construct the control inputs, one cannot solve \optref to compute the \Linf-gain bound if the output variables,~$z$, depend on some control inputs.   						  
\section{Application to Feedback Motion Planning}\label{sec:app-mp}
Thanks to the certificate tubes in \cref{eq:tube_zhat}, the RCCM controller presented in \cref{sec:RCCM} can be conveniently incorporated as a low-level tracking or ancillary controller into a feedback motion planning or nonlinear tube MPC framework. We demonstrate an application to the former in this section. The core idea is to compute nominal motion plans $(\xstar, \ustar)$ using the nominal dynamics \eqref{eq:dynamics-nom} and  {\it tightened} constraints.  
Denote the tubes for $x-\xstar$ and $u-\ustar$ obtained through solving \optref~in \cref{sec:sub-opt_refine_u_x_tubes} as $\tilde \Omega_x \trieq \{\tilde x\in \mbR^n: \norm{\tilde x} \leq \alpha_x \bar w\}$ and $\tilde \Omega_u \trieq \{\tilde u\in \mbR^m: \norm{\tilde u} \leq \alpha_u \bar w \}$, where $\alpha_x$ and $\alpha_u$ are the universal \Linf-gain bounds for the states and control inputs, respectively, and $\bar w$ is bound on the disturbances, i.e., $\linfnorm{w-\wstar} \leq \bar w$. Then, the tightened constraints are given by 
\begin{subequations}\label{eq:tight_csts}
\begin{align}
    \xstar(\cdot) \in \bar \mcX \trieq \mcx \ominus \tilde \Omega_x  \label{eq:tight_csts_x},\\
     \ustar(\cdot) \in \bar \mcU \trieq \mcU \ominus \tilde \Omega_u \label{eq:tight_csts_u},
\end{align}
\end{subequations}
where $\mcU$ represents the control constraints, and $\ominus$ denotes the Minkowski set difference. 
One can simply use the tightened constraints in \eqref{eq:tight_csts} and the nominal dynamics \eqref{eq:dynamics-nom} to plan a target trajectory. Then, with the proposed RCCM controller, the actual states and inputs are guaranteed to stay in $\mcX$ and $\mcU$, respectively, in the presence of disturbances bounded by $\bar w$.
\begin{remark}
Dependent on the tasks, one may want to focus on some particular states when designing the RCCM controller through solving \optrccm. For instance, for motion planning with obstacle-avoidance requirements, one may want to focus on minimizing the tube size for position states. This often leads to tight tubes for position states, enabling planning more aggressive yet safe motions, as demonstrated in \cref{sec:simulations}.
\end{remark}

\section{Comparisons with
an Existing CCM-based Approach}\label{sec:compare-with-ccm}
In \cite{singh2019robust}, for the same system \eqref{eq:dynamics} considered here, the authors designed a tracking controller based on CCM {\it without considering disturbances} and then derived a {\it tube} where the actual states are guaranteed to stay {\it in the presence of disturbances} using input-to-state stability (ISS) analysis. In comparison, our method {\it explicitly incorporates disturbance rejection property} in designing the RCCM controller and produces tubes for both states and inputs {\it together} with the controller (if we include the tube refining process in \cref{sec:sub-opt_refine_u_x_tubes} as a part of the controller design process). In this section, under mild assumptions, we will prove that the tube yielded by our method is {\it tighter} than that from applying the idea of  \cite{singh2019robust}. To be consistent with the problem setting in \cite{singh2019robust}, for this section, we set $w^\star \equiv 0$ in defining the nominal  (i.e., un-disturbed) system \eqref{eq:dynamics-nom}, which leads to the nominal dynamics:  
\begin{equation}\label{eq:dynamics-nom-wstar=0}
   \dotxnom = f(\xnom)+B(\xnom) \ustar.
\end{equation}
The main technical ideas from \cite{singh2019robust} (mainly related to Theorem 3.5, Lemma~3.7 and Section~4.2 of \cite{singh2019robust}) can be summarized as: (1) searching a (dual) CCM metric, $\hat W$, for the nominal system \eqref{eq:dynamics-nom-wstar=0}, which yields a nonlinear controller guaranteeing the incremental stability of the nominal close-loop system; (2) deriving a tube to quantify the actual state in the presence of disturbances, i.e., subject to the dynamics \eqref{eq:dynamics}, based on ISS analysis. 

Unlike our approach, in \cite{singh2019robust} the search of the CCM metric is {\it not} jointly done with search of a matrix function (i.e., $Y(x)$ in Theorem~\ref{them:synthesis-universal-Linf-gain}, used to construct a differential feedback controller). Instead, \cite{singh2019robust} uses a min-norm type control law computed  using only the CCM metric. To facilitate a rigorous comparison, we slightly modify the condition for the CCM metric search to include another matrix function (analogous to $Y(x)$ in Theorem~\ref{them:synthesis-universal-Linf-gain}). Indeed, a joint search of a CCM metric $W(x)$ and a matrix function $Y(x)$  is adopted in \cite{manchester2017control}, which \cite{singh2019robust} builds upon. Such modification only influences the control signal determination, and does {\it not} change the {\it essential ideas} of \cite{singh2019robust}. With such modifications, the main results of  \cite{singh2019robust} can be summarized in the following lemma using the notations of this paper.
\begin{lemma}\label{lemma:ccm-tube}(\hspace{-0.1mm}\cite{singh2019robust})
For the nominal system \cref{eq:dynamics-nom-wstar=0}, assume there exists a metric $\hat W(\xnom)$, a matrix function $\hat Y(\xnom)$ and a constant $\hat \lambda>0$ satisfying
\begin{equation}\label{eq:synthesis-condition-ccm}
    -\dot{\hat W} + \langle \hat A \hat W+B \hat Y \rangle + 2\hat\lambda \hat W \leq  0,
\end{equation}
where $\hat A\trieq \frac{\partial f}{\partial \xstar} + \sum\limits_{i = 1}^m {\frac{{\partial {b_i}}}{{\partial x}}} {\ustar_i}$, $\dot {\hat W} = \sum _{i=1}^{n} \frac{\partial \hat W(\xnom)}{\partial \xnom_i} \dotxnom_i$. 
Furthermore, $\hat W(\xnom)$ is uniformly bounded, i.e. $\ubar{\beta } I_n \leq\hat W(\xnom) \leq\bar \beta I_n$ with $\bar \beta \geq \ubar \beta>0$, for all $\xnom\in \mcX$. Then, for the perturbed system \cref{eq:dynamics} under the controller \eqref{eq:RCCM-controller}  with $W =\hat W$ and $Y = \hat Y$, if $x(0)=x^\star(0)$, we have
\begin{equation}\label{eq:Linf-gain-inequ-ccm}
    \linfnormtrucI{x - x^\star}{[0,T]}^2  \leq\hat\alpha^2 \linfnormtrucI{w}{[0,T]}^2, 
\end{equation}
where %$\norm{w(t)}$ \bar w \triangleq \sup_{t\geq 0 }
 \begin{equation}\label{eq:alpha-hat-defn}
 \hat \alpha  \triangleq \frac{1}{{\hat \lambda }}\sqrt {{\bar \beta }/{\ubar \beta }} \sup_{x \in \mcX} \bar \sigma ({B_w}(x)),
 \end{equation}
 with $\bar\sigma(\cdot)$ denoting the largest singular value. 
\end{lemma}
% \begin{remark}\label{}

% \end{remark}

We also need the following assumption. 
\begin{assumption}\label{assump:What-u-w-independent}
The metric $\hat W$ in \cref{eq:synthesis-condition-ccm} satisfies both of the following conditions:
\begin{itemize}[leftmargin = 10 mm]
\item[\cII] For each $i=1,\dots,m$, $\partial_{b_i} \hat W - \left\langle \frac{\partial b_i}{\partial x}\hat W\right\rangle =0$.
\item[\cIII] For each $i=1,\dots,p$, $\partial_{b_{w,i}}\hat W - \left\langle \frac{\partial b_{w,i}}{\partial x}\hat W\right\rangle =0$.
%     \item[\cII] $\hat W$ is independent of any state $\xstar_i$, whose derivative is dependent on $\ustar$;
%     \item [\cIII] $\hat W$ is independent of state $\xstar_i$ if the derivative of $x_i$ in \eqref{eq:dynamics} is dependent on $w$. 
\end{itemize}
\end{assumption}
Condition \cII~is similar to condition \cI, and is also imposed in \cite{singh2019robust} to simplify the verification of \cref{eq:synthesis-condition-ccm} and get a controller with a simple differential feedback form (see \cite[III.A]{manchester2017control}). Condition \cIII~states that each $b_{w,i}$ forms a Killing vector for $\hat W$, which essentially ensures that the condition \eqref{eq:synthesis-condition-ccm}, evaluated using the perturbed dynamics (i.e., replacing $\hat A$ in \cref{eq:synthesis-condition-ccm} with $A$ below \eqref{eq:diff_dynamics}), does {\it not} depend on $w$. Now we are ready to build a connection between the CCM-based approach in \cite{singh2019robust} and our approach. 
% In particular, if in \eqref{eq:dynamics}, $B$ and $B_w$ are in the form of $[0, I_{m_1}, 0]^\top$ and $[0, 0, I_{m_2}]^\top$, respectively, where $I_{m_i}$ ($i=1,2$) is $m_i$ by $m_i$ identity matrix, Conditions (A1) and (A2) require that $\hat W$ must not depend on the last $(m_1 + m_2)$ state variables. Under \cref{assump:What-u-w-independent}, we have $\dot {\hat W} = \sum _{i=1}^{n} \frac{\partial \hat W(\xnom)}{\partial \xnom_i}{f_i(\xnom)}$.

\begin{lemma}\label{lemma:nominal-ineq-imply-PPG-ineq}
Assume there exists a metric $\hat W(x)$, a matrix function $\hat Y(x)$, and a constant $\hat \lambda>0$ satisfying \eqref{eq:synthesis-condition-ccm} and \cref{assump:What-u-w-independent}. Then, \cref{eq:diff-Linf-synthesis-1,eq:diff-Linf-synthesis-2} with $C=I_n$ and $D=0$ (corresponding to $g(x,u) = x$) can be satisfied with 
\begin{equation}\label{eq:W-What-lambda-lambdahat-csts}
 \hspace{-2mm} W(x) = a\hat W(x),\ Y(x) = a\hat Y(x),\  \lambda=\hat \lambda,\ \alpha =\mu = \hat \alpha,
\end{equation}
where $a \!\trieq\! {\sup_{x \in \mcX} \bar \sigma ({B_w}(x))}/\!{\sqrt{\bar \beta \ubar \beta}}$, and $\hat \alpha$ is defined in \eqref{eq:alpha-hat-defn}.
% where $a=\frac{\sup_{x \in \mcX} \bar \sigma ({B_w}(x))}{\sqrt{\bar \beta \ubar \beta}}$,  satisfy
\end{lemma}
\begin{proof}

Under the constraint of $\mu>0$, by applying the Schur complement, we can rewrite \eqref{eq:diff-Linf-synthesis-1} as
$
\langle AW+BY \rangle - \dot W + \lambda W + \frac{1}{\mu} B_w B_w^\top\leq 0.$
Due to \eqref{eq:W-What-lambda-lambdahat-csts} and Conditions \cII~and \cIII, the preceding inequality is equivalent to  $a \langle A \hat W+B \hat Y \rangle -  a \dot {\hat W} + a \hat \lambda W + \frac{1}{\hat \alpha} B_w B_w^\top\leq 0$, which, due to \cref{eq:synthesis-condition-ccm}, will hold if 
$
    B_w B_w^\top \leq a \hat \alpha \hat \lambda \hat W, 
$
or  equivalently (considering the definitions of $a$ below \cref{eq:W-What-lambda-lambdahat-csts} and $\hat \alpha$ in \cref{eq:alpha-hat-defn})
\begin{equation}\label{eq:beta-Bw-W-ineq}
    \ubar{\beta } {B_w}B_w^ \top   \leq {\left( {\mathop {\sup }_{x \in \mcX} \bar \sigma ({B_w}(x))} \right)^2}\hat W   
\end{equation}
holds. Since ${B_w}B_w^ \top  \leq {\left( {\mathop {\sup }_{x \in \mcX} \bar \sigma ({B_w}(x))} \right)^2}I_n$ and $\hat W \geq \ubar\beta I_n$, \eqref{eq:beta-Bw-W-ineq} holds, and therefore, \eqref{eq:diff-Linf-synthesis-1} holds.

On the other hand, by applying the Schur complement, \eqref{eq:diff-Linf-synthesis-2} with $C=I_n$ and $D=0$ is equivalent to
\begin{subequations}\label{eq:alpha-mu-W-together}
\begin{align} 
    \alpha\geq \mu, \label{eq:alpha-mu} \\ 
    \lambda W-\frac{1}{\alpha} WW \geq 0. \label{eq:W-WW-ineq}
\end{align}
\end{subequations}
 Equation \eqref{eq:alpha-mu} trivially holds due to \eqref{eq:W-What-lambda-lambdahat-csts}. Rewrite $W(x)$ as $W(x) = \Lambda(x)^\top \Lambda(x)$ with $\Lambda(x)$ being non-singular  for any $x\in \mcX$. Multiplying the left hand side of \eqref{eq:W-WW-ineq} by $\Lambda(x)^{-\top}$ and its transpose from the left and right, respectively, leads to $
    \Lambda \Lambda^\top \leq \frac{\alpha\lambda}{a} I_n.
$
Due to \eqref{eq:W-What-lambda-lambdahat-csts}, the preceding inequality is equivalent to $\Lambda \Lambda^\top \leq \bar \beta I_n$, which holds since $\bar\lambda(\Lambda(x)\Lambda^\top(x))=  \bar\lambda(\Lambda^\top(x)\Lambda(x)) = \bar\lambda(W(x))$  for any $ x$ with $\bar\lambda$ denoting the largest eigenvalue. Therefore, \eqref{eq:diff-Linf-synthesis-2} holds. The proof is complete. \qed

\end{proof}

According to \cref{lemma:nominal-ineq-imply-PPG-ineq}, if we can find matrices $\hat{W}$ and $\hat{Y}$ and constants $\hat \lambda$ satisfying the inequality \eqref{eq:synthesis-condition-ccm}, which guarantees the contraction of the nominal closed-loop system and ensures an \Linf-gain bound $\hat \alpha$ from disturbances to states, we can obtain the same \Linf-gain bound using our approach (\cref{them:synthesis-universal-Linf-gain}), if  we choose $W(x)$ and $Y(x)$ in \cref{eq:diff-Linf-synthesis-1,eq:diff-Linf-synthesis-2} to be scaled versions of $\hat W(x)$ and $\hat Y(x)$ in \eqref{eq:synthesis-condition-ccm}, i.e., enforcing the constraints in  \eqref{eq:W-What-lambda-lambdahat-csts}. However, if we relax such constraints in the optimization problem 
\optrccm, we are guaranteed to obtain a less conservative bound $\alpha$, i.e., $\alpha \leq \hat \alpha$.  This observation is summarized in the following theorem with the straightforward proof omitted.
\begin{theorem}\label{them:ccm-rccm-comparison}
Assume there exist a metric $\hat W(x)$, a matrix function $\hat Y(x)$, and a constant $\hat \lambda>0$  satisfying \eqref{eq:synthesis-condition-ccm} and \cref{assump:What-u-w-independent}. Then, we can always find $W(x)$, $Y(x)$, $\lambda>0$, $\mu>0$ and $\alpha\leq \hat \alpha$ satisfying \cref{eq:diff-Linf-synthesis-1,eq:diff-Linf-synthesis-2} with $C=I_n$ and $D=0$, where $\hat \alpha$ is defined in \eqref{eq:alpha-hat-defn}.
% Optimization
% \begin{align}
%     & \min\limits_{\lambda, \mu, W(x), Y(x)} \alpha \\
%   \textup{ subject to } & \eqref{eq:diff-Linf-synthesis-1},\  \eqref{eq:diff-Linf-synthesis-2}\textup{ with } C=I \textup{ and } D=0.  
% \end{align}
\end{theorem}
% \begin{proof}
% Lemma~\ref{lemma:nominal-ineq-imply-PPG-ineq} has shown that we can at least obtain a same PPG bound, i.e., $\alpha=\hat \alpha$, if we choose $W(x)$ and $Y(x)$ in \cref{eq:diff-Linf-synthesis-1,eq:diff-Linf-synthesis-2} to be scaled versions of $\hat W(x)$ and $\hat Y(x)$ in \eqref{eq:synthesis-condition-ccm}, i.e., enforcing the constraints $ W(x) = a\hat W(x)$ and  $ Y(x) = a\hat Y(x)$. If we relax such constraints as in the optimization problem, we are guaranteed to obtain a less conservative bound $\alpha$, i.e., $\alpha \leq \hat \alpha$. 
% \end{proof}
\begin{remark}
\cref{them:ccm-rccm-comparison} indicates that our proposed RCCM approach is {\it guaranteed} to yield a {\it tighter} tube for the actual states than the CCM approach in \cite{singh2019robust}, under \cref{assump:What-u-w-independent}. 
\end{remark}

\section{Simulation Results}\label{sec:simulations}
In this section, we apply the proposed approach to a planar VTOL vehicle (illustrated in Fig.~\ref{fig:illustration}) \black{and a 3D quadrotor and perform extensive comparisons with the CCM-based approach in \cite{singh2019robust}. All the subsequent computations and simulations were done in Matlab R2021a\footnote{Matlab codes are available at \rurl{github.com/boranzhao/robust_ccm_tube}.}.} A video for visualizing the simulation results is available at \rurl{youtu.be/mrN5iQo7NxE}.
\subsection{Planar VTOL vehicle}\label{sec:sub-pvtol}
The state vector is defined as $x = [p_x, p_z, \phi, v_x,v_z,\dot \phi]^\top$, where $p = [p_x,p_z]^\top$ is the position in $x$ and $z$ directions, respectively, $v_x$ and $v_z$ are the slip velocity (lateral) and the velocity along the thrust axis in the body frame of the vehicle, $\phi$ is the angle between the $x$ direction of the body frame and the $x$ direction of the inertia frame. The input vector $u=[u_1,u_2]$ contains the thrust force produced by each of the two propellers.  The dynamics of the vehicle are given by
{\setlength{\arraycolsep}{1pt}
{\begin{equation*}
\small
   \begin{bmatrix}
    \dot p_x \\
    \dot p_z \\
    \dot \phi \\
    \dot v_x \\
    \dot v_z \\
    \ddot \phi 
    \end{bmatrix} =
    \begin{bmatrix}
    v_x \cos(\phi)-v_z \sin(\phi)\\
    v_x\sin(\phi) + v_z \cos(\phi) \\
    \dot \phi \\
    v_z\dot \phi - g\sin(\phi) \\
    -v_x\dot \phi - g\cos(\phi) \\
    0  
    \end{bmatrix} +
    \begin{bmatrix}
    0 & 0 \\
    0 & 0 \\
    0 & 0 \\
    0 & 0 \\
    \frac{1}{m} & \frac{1}{m} \\
    \frac{l}{J} & -\frac{l}{J}
    \end{bmatrix}u + \begin{bmatrix}
     0 \\
     0 \\
     0\\
     \cos(\phi) \\
     -\sin(\phi)\\
      0
    \end{bmatrix}w,
\end{equation*}}}where $m$ and $J$ denote the mass and moment of inertia about the out-of-plane axis and $l$ is the distance between each of the propellers and the vehicle center, and $w$ denotes the disturbance in $x$ direction of the inertia frame. Following \cite{singh2019robust}, the parameters were set as $m=0.486$ kg, $J=0.00383~\textup{Kg m}^2$, and $l=0.25$ m.  
				   
\subsubsection{Computation of CCM/RCCM and associated tubes}
\black{We parameterized both the RCCM $W$ and the CCM $\hat W$ as polynomial matrices in  $(\phi, v_x)$ with up to degree 4 monomials}. When searching for CCM/RCCM, we also imposed the following bounds: $(v_x,v_z)\in [-2,2]\times[1,1]$ m/s, and $(\phi,\dot \phi) \in [-60^\circ, 60^\circ]\times [-60,60]^\circ/s$, which can be concatenated as the vector constraint $h(x)\geq 0$.  For a fair comparison of the proposed RCCM-based approach and the CCM-based approach in \cite{singh2019robust}, we used same parameters when searching CCM and RCCM whenever possible. For instance, we used the same basis functions for parameterizing $W$ and $\hat W$ when applying the SOS techniques to solve the optimization problems, and imposed the same lower bound for $W$ and $\hat W$: $W\geq 0.01 I_6$ and $\hat W \geq 0.01 I_6$.

We first considered the optimization of the tube size for all the states, on which \cite{singh2019robust} is focused.  For simplicity, we did not use weights for the states. For RCCM synthesis, we included a penalty for large control efforts when solving \optrccm~by setting $g(x,u) = [x^\top,u^\top]^\top$, and denote the resulting controller as \textbf{RCCM}.  Additionally, we designed another RCCM controller with a focus on optimizing the tubes for the {\it position states} and inputs, denoted as {\bf RCCM-P},  by setting $g(x,u) = [p_x, p_z,u^\top]^\top$. We denote the controller designed using the CCM approach in \cite{singh2019robust} as {\bf CCM}. 

We considered a cross-wind disturbance along the $x$ direction of the inertia frame with effective acceleration up to $1$ m/s (i.e., $\bar w = 1$), which is {\it 10 times} as large as the disturbance considered in \cite{singh2019robust}.   We swept through a range of values for $\lambda$ (setting $\hat \lambda=\lambda$) and solved the \optrccm~in \cref{sec:sub-opt-rccm} to search for the RCCM and the optimization problem in \cite[Section~4.2]{singh2019robust} to search for the CCM, 
using SOS techniques with YALMIP \cite{YALMIP} and Mosek solver \cite{andersen2000mosek}. After obtaining the RCCM, we further solved \optref~in \cref{sec:sub-opt_refine_u_x_tubes} by gridding the state space to get refined tubes for different variables. The results are shown in Fig.~\ref{fig:tube_size_lambda}. 
According to the top plot, while both controllers focused on optimizing the tube size for {\it all states} without using weights, RCCM yielded a much smaller tube than CCM. 
%\al{RCCM-P is much smaller than CCM right? if yes then this sentence requires a fix} Pan: the preceding sentence talks about the results when considering all the states in the optimization. 
RCCM-P yielded a tube of similar size for all states compared to CCM, which came as no surprise since RCCM-P focused on minimizing the tube size for {\it position states only}, i.e., $(p_x, p_z)$. From the middle and bottom plots, one can see that RCCM-P yielded much smaller tubes for both position states and inputs than RCCM, which further outperforms CCM by a large margin. 
\begin{figure}
    \centering
    \includegraphics[width=1\columnwidth]{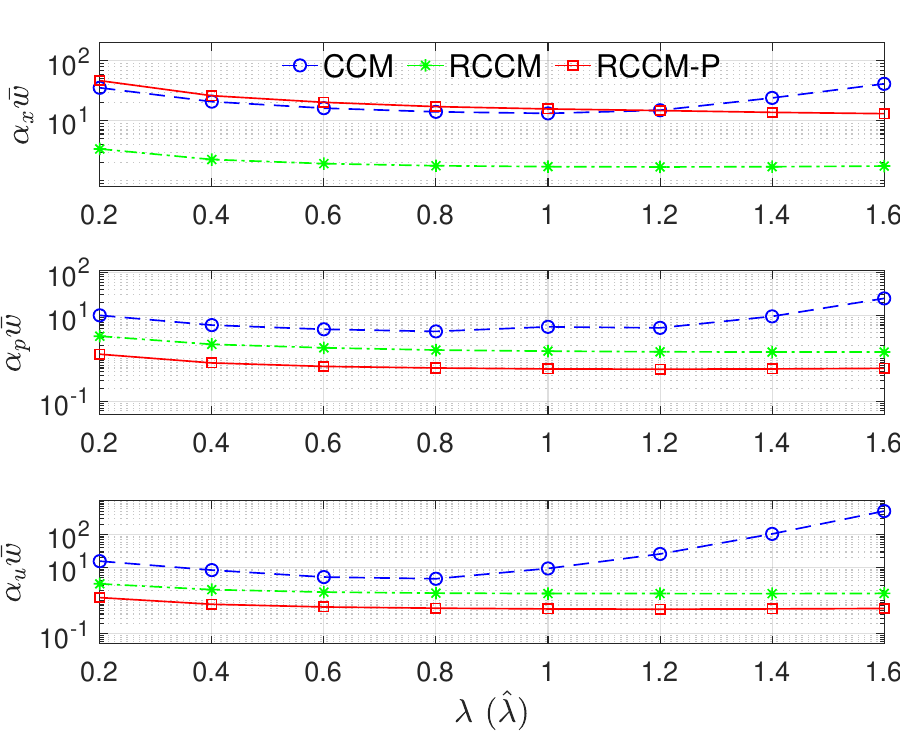}
	    \caption{Tube size for all states (top), position states (middle) and inputs (bottom) versus $\lambda$ ($\hat \lambda$) value in the presence of disturbances bounded by $\bar w = 1$}
    \label{fig:tube_size_lambda}
\vspace{-3mm}
\end{figure}
For subsequent tests and simulations, we selected a best $\lambda$ value for each of the three controllers in terms of tube size for $(p_x, p_z)$, since the vehicle position is of more importance in tasks with collision-avoidance requirements. The best values for CCM, RCCM, RCCM-P are determined to be 0.8, 1.4 and 1.2, respectively. Figure~\ref{fig:tube} depicts the input tube, and the projection of the state tube onto different planes, yielded by each of the three controllers with the best $\lambda$ value. It is no surprise that RCCM-P, while yielding much smaller tubes for $(p_x,p_z)$ and inputs, results in relatively larger tubes for $(v_x,v_z)$ and $(\phi,\dot\phi)$. 
\begin{figure}
    \centering
    \includegraphics[width=1\columnwidth]{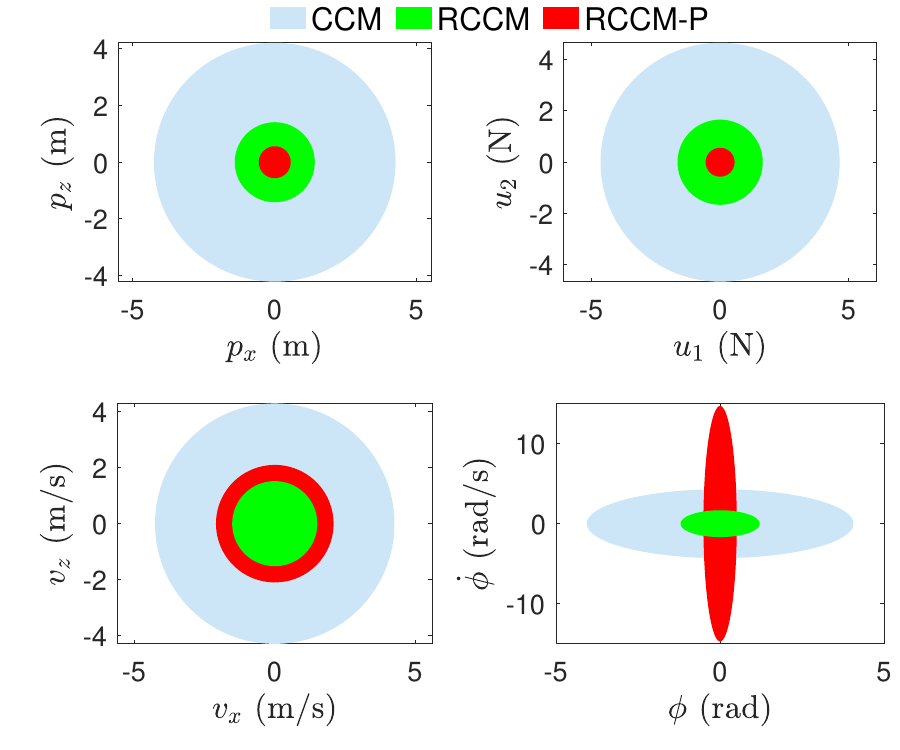}
    \caption{Projection of state and input tubes under wind disturbances with effective acceleration up to 1 m/s$^2$}
    \label{fig:tube}
\vspace{-3mm}							 
\end{figure}
\subsubsection{Trajectory tracking and verification of tubes}\label{sec:traj_test_rsts}
To test the trajectory tracking performance of the three controllers in scenarios and evaluate the  conservatism with the derived tubes, we considered a task of navigation from the origin to target point $(10,10)$. We first planed a nominal trajectory with the objective of minimal force and minimal travel time, using OptimTraj \cite{kelly2017intro-traj-opt}, where the state constraint $h(x)\geq 0$, used in searching for CCM/RCCM, was enforced. With the nominal state and input trajectories, we simulated the performance of controllers in the presence of a wind disturbance, artificially simulated by $w(t) = 0.8+0.2\sin(2\pi t/10)$. OPTI \cite{Currie12opti} and Matlab \texttt{fmincon} solvers were used to solve the geodesic optimization problem at each sampling instant for all the three controllers (see \cref{sec:sub-control-law} for details). With Matlab 2021a running on a PC with Intel i7-4790 CPU and 16 GB RAM and generated C codes for evaluating the cost function and gradient, it took roughly $20\sim30$ milliseconds to solve the optimization problem for computing the geodesic once. 
The results of the position trajectories along with the tubes projected to the $(p_x,p_z)$ plane are shown in Fig.~\ref{fig:traj}. First, it is clear that the actual trajectory under each controller always stays in the associated tube. Second, in terms of tracking performance, RCCM-P and CCM perform the best and worst, respectively. 
\begin{figure}
    \centering
    \includegraphics[width=1\columnwidth]{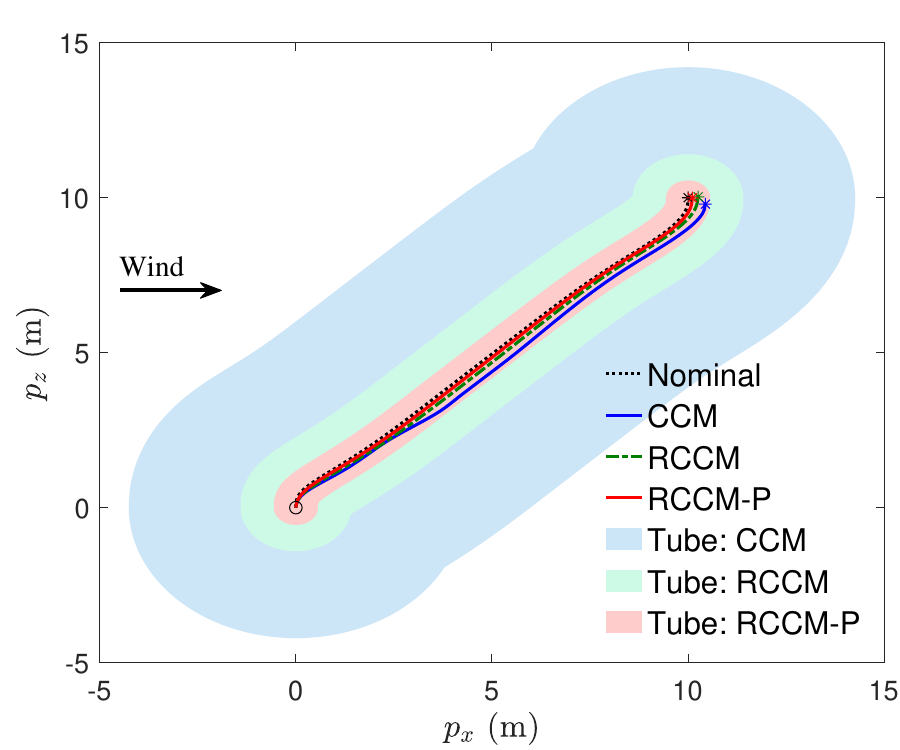}
     \includegraphics[width=1\columnwidth]{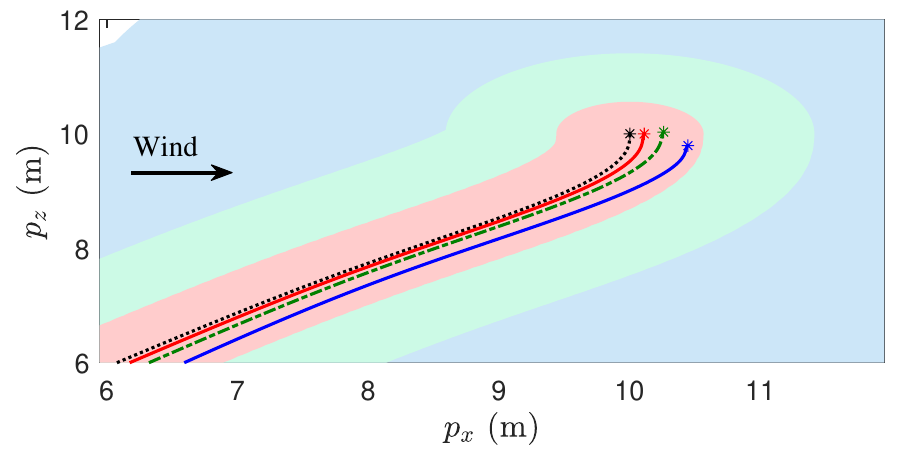}
    \caption{Tracking of a nominal trajectory by different controllers: full (top) and zoomed-in (bottom) view}
    \label{fig:traj}
    \vspace{-4mm}
\end{figure}
\subsubsection{Feedback motion planning and tracking in the presence of obstacles}\label{sec:traj_w_obs_test_rsts}
We now consider a joint trajectory planning and tracking problem for the same task considered in \cref{sec:traj_test_rsts} but in the presence of obstacles, illustrated as black circles in \cref{fig:traj_w_obs}. We followed the feedback motion planning framework and incorporated the tubes for both position states and inputs when planning the trajectory. For simplicity, we ignored the tubes for other states (i.e., $v_x$, $v_z$, $\phi$, $\dot \phi$) in the planning.  The planned trajectory and tube associated with each controller are denoted by a black dotted line and a shaded area in \cref{fig:traj_w_obs}. As expected, the trajectory optimizer found different trajectories for the three controllers due to different tube sizes. The travel time associated with the planned trajectories under CCM, RCCM, and RCCM-P are 18.0, 11.8 and 10.1 seconds, respectively, with RCCM-P yielding the shortest travel time. The actual trajectories in the presence of the disturbances are also included in \cref{fig:traj_w_obs}. It is clear that the actual trajectory under each of three controllers always stays in the tube around its associated nominal trajectory and remains collision-free. Once again, RCCM-P yielded the smallest tracking error. 
\begin{figure}
    \centering
    \includegraphics[width=1\columnwidth]{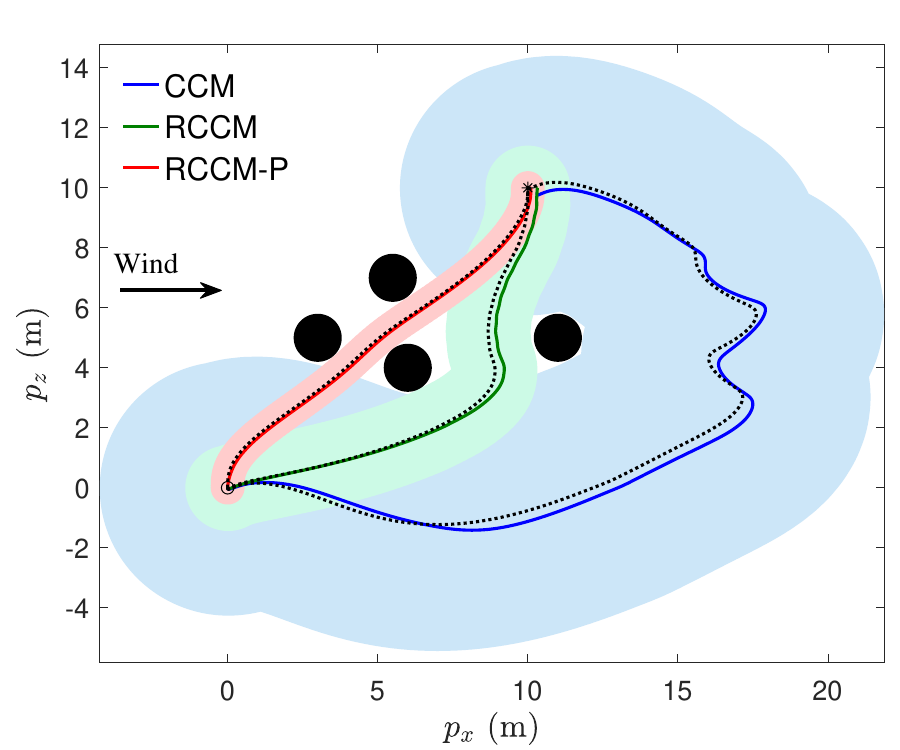}
     \includegraphics[width=1\columnwidth]{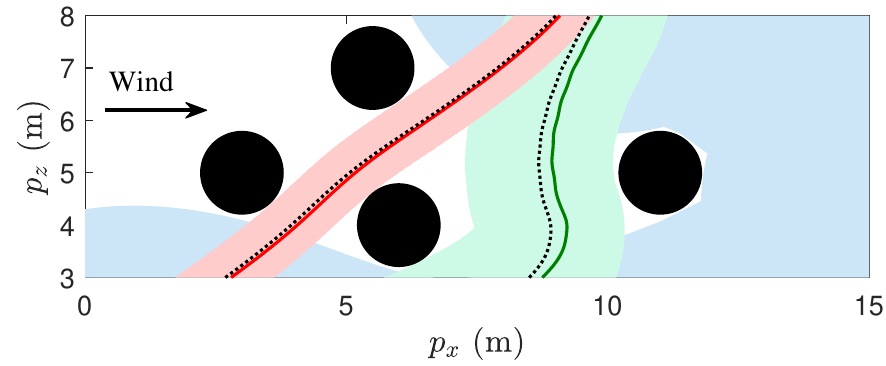}
    \caption{Planning and tracking of a nominal trajectory by different controllers incorporating safety tubes: full (top) and zoomed-in (bottom) view. Dotted lines denote planned trajectories. Shaded areas denote the tubes for the position states.}
    \label{fig:traj_w_obs}	
    \vspace{-3mm}
\end{figure}

\begin{figure}[h]
    \centering
    \includegraphics[width=1\columnwidth]{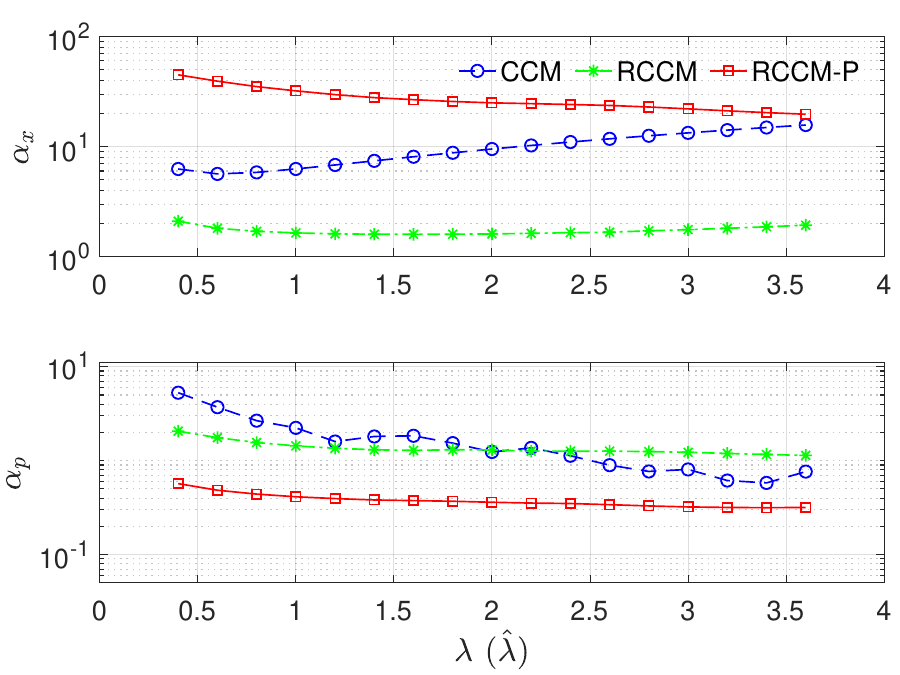}
	    \caption{Tube size gain for all states (top) and position states (bottom) versus $\lambda$ ($\hat \lambda$) value}
    \label{fig:tube_size_lambda-quad}
    \vspace{-3mm}
\end{figure}

\subsection{3D quadrotor}
The 3D quadrotor model is taken from \cite{singh2019robust} and has the state-space representation given by $x=[p_x,p_y,p_z,\dot p_x, \dot p_y, \dot p_z,\tau, \phi, \theta, \psi]^\top$, where the position $p=[p_x,p_y,p_z]^\top \in \mbR^3$ and corresponding velocities are expressed in the global inertial (vertical axis pointing down) frame.  Adopting the North-East-Down frame convention for the quadrotor body and the XYZ Euler-angle rotation sequence, the attitude (roll, pitch, yaw) is parameterized as
$(\phi, \theta,\psi)$ and $\tau >0$ is the total (normalized by mass) thrust generated by the four rotors. For controller design, we consider  $u \trieq [\dot \tau, \dot \phi,\dot \theta]^\top$ as the control input. The actual implementation embeds the $\dot \tau$ term within an integrator, and the resulting thrust and angular velocity reference (after being converted to body rate reference) are passed to a lower-level controller that is assumed to operate at a much faster
time-scale. Given this parameterization, the dynamics of the quadrotor may be written as
\begin{equation}
  \hspace{-2mm}  \begin{bmatrix}
    \ddot p_x \\ \ddot p_y \\ \ddot p_z
    \end{bmatrix}
    = ge_3-\tau \hat b_z + w = 
    \begin{bmatrix}
    -\tau \sin(\theta) +w_1\\ \tau \cos(\theta) \sin(\phi) +w_2\\ g - \tau \cos(\theta)\cos(\phi) + w_3
    \end{bmatrix},
\end{equation}
where $g$ is the local gravitational acceleration, $e_3 = [0,0,1]^\top$, $\hat b_z$ is the body-frame $z$-axis, and $w=[w_1,w_2,w_3]^\top\in\mbR^3$ denotes the disturbance. The dynamics of $(\tau, \phi, \theta, \psi)$ reduce trivially to first-order integrators. We impose the following bounds: $(\phi,\theta)\!\in\! [-60^{\circ},60 ^{\circ}]^2$, $\tau \!\in \![0.5,2]g$, $(\dot \phi, \dot \theta)\!\in\! [-180^{\circ}/\textrm{s},180 ^{\circ}/\textrm{s}]^2$, and $\dot\tau \!\in\! [-5,5]g/\textrm{s}$, which are sufficient for fairly aggressive maneuvers. Since yaw control is not a focus here, we simply set $\dot \psi =0$.
\subsubsection{Computation of CCM/RCCM and associated tubes}
We parameterize both the RCCM $W$ and the CCM $\hat W$ as polynomial matrices in $(\tau,\phi,\theta)$ with up to degree 3 monomials. Additionally, the top left $6\times 6$ block of $\hat W$ was imposed to be constant, i.e., independent of  $(\tau,\phi,\theta)$ to ensure that the resulting synthesis condition did not depend on $u$ (see \cite{singh2019robust} for details). The RCCM synthesis condition \cref{eq:diff-Linf-synthesis-1}, however, depends on $(\dot \tau, \dot \phi, \dot \theta)$, which is why we impose the bounds on $(\dot \tau, \dot \phi, \dot \theta)$ mentioned above. Similar to \cref{sec:sub-pvtol}, we first consider the optimization of the tube size for all the states. For simplicity, we do not use weights for the states, which can also be considered as using equal weights for all the states. For RCCM synthesis, we include a penalty for large control efforts required when solving \optrccm~by setting $g(x,u) = [x^\top\!, 0.02\dot \tau , 0.05 \dot \phi, 0.05\dot \theta]^\top$, and denote the resulting controller as \textbf{RCCM}.  Additionally, we design another RCCM controller with a focus on optimizing the tubes for the {\it position states}. For this, we set $g(x,u) = [p_x, p_y,p_z,0.02\dot \tau, 0.05 \dot \phi, 0.05\dot \theta]^\top$ and denote the resulting controller as {\bf RCCM-P}. We denote the controller designed using the CCM-based approach in \cite{singh2019robust} as {\bf CCM}. Through numerical experimentation, we found that imposing the lower bound constraint  $W\geq 0.01 I_9$ yielded good performance for searching $W$, while imposing the constraint $\hat W\geq I_9$ yielded good performance when searching $\hat W$. We swept through a range of values from $0.4$ to $3.6$ for $\lambda$ (setting $\hat \lambda=\lambda$) and solved the \optrccm~in \cref{sec:sub-opt-rccm} to search for the RCCM and the optimization problem in \cite[Section~4.2]{singh2019robust} to search for the CCM. We first tried the SOS technique used in \cref{sec:sub-pvtol} to solve the involved optimization problems, but found it was not reliable especially for RCCM synthesis, taking notoriously long time while still yielding unsatisfactory results. Therefore, we eventually chose to grid the set of $(\tau,\phi,\theta)$ (and additionally $(\dot \tau,\dot \phi,\dot \theta)$ for RCCM search), and solved the resulting optimization problem with a finite number of LMIs with YALMIP \cite{YALMIP} and Mosek solver \cite{andersen2000mosek}.  
After obtaining the RCCM, we further solved \optref~in \cref{sec:sub-opt_refine_u_x_tubes} using the gridding technique to get refined tubes for different variables. The results in terms of tube size gain are shown in Fig.~\ref{fig:tube_size_lambda-quad}. 
As shown in the top plot, while both controllers focused on optimizing the tube size for {\it all states} without using weights, RCCM yielded a much smaller tube than CCM. 
%\al{RCCM-P is much smaller than CCM right? if yes then this sentence requires a fix} Pan: the preceding sentence talks about the results when considering all the states in the optimization. 
RCCM-P yielded a tube of similar size for all states compared to CCM, which came as no surprise since RCCM-P focused on minimizing the tube size for {\it position states only}. From the bottom plots, one can see that RCCM-P yielded much smaller tubes for both position states and inputs than RCCM, which further outperforms CCM when $\lambda$ is less than 2.
Note that \cref{assump:What-u-w-independent} does not hold anymore for this example, as the condition \cII~cannot be satisfied. Therefore, \cref{them:ccm-rccm-comparison} does not hold, which indicates that we cannot {\it guarantee} that RCCM will yield tighter tubes for all the states than CCM. %However, the numerical experiments empirically showed that this is still true.  
\begin{figure}[h]
    \centering
        \vspace{-2mm}
    \includegraphics[width=1\columnwidth]{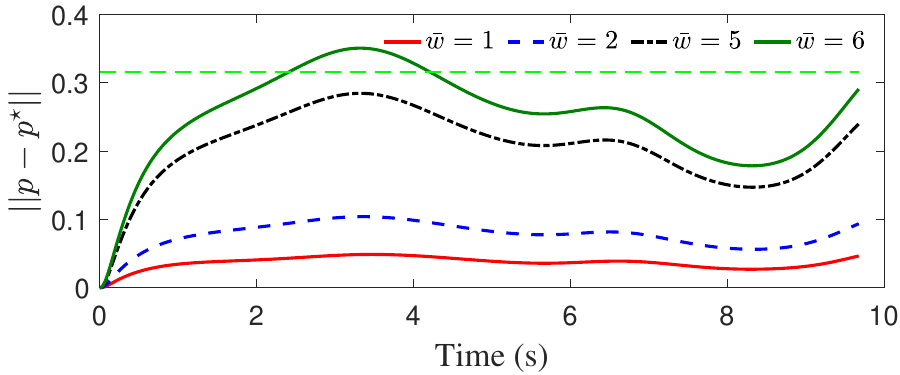}
    \caption{Tracking error for the position states under RCCM-P. The green dashed line denotes the theoretical bound associated with $\bar w= 1$.}
    \label{fig:traj-err-quad}
    \vspace{-3mm}
\end{figure}
\begin{figure}[h]
    \centering
        \vspace{-2mm}
    \includegraphics[width=1\columnwidth]{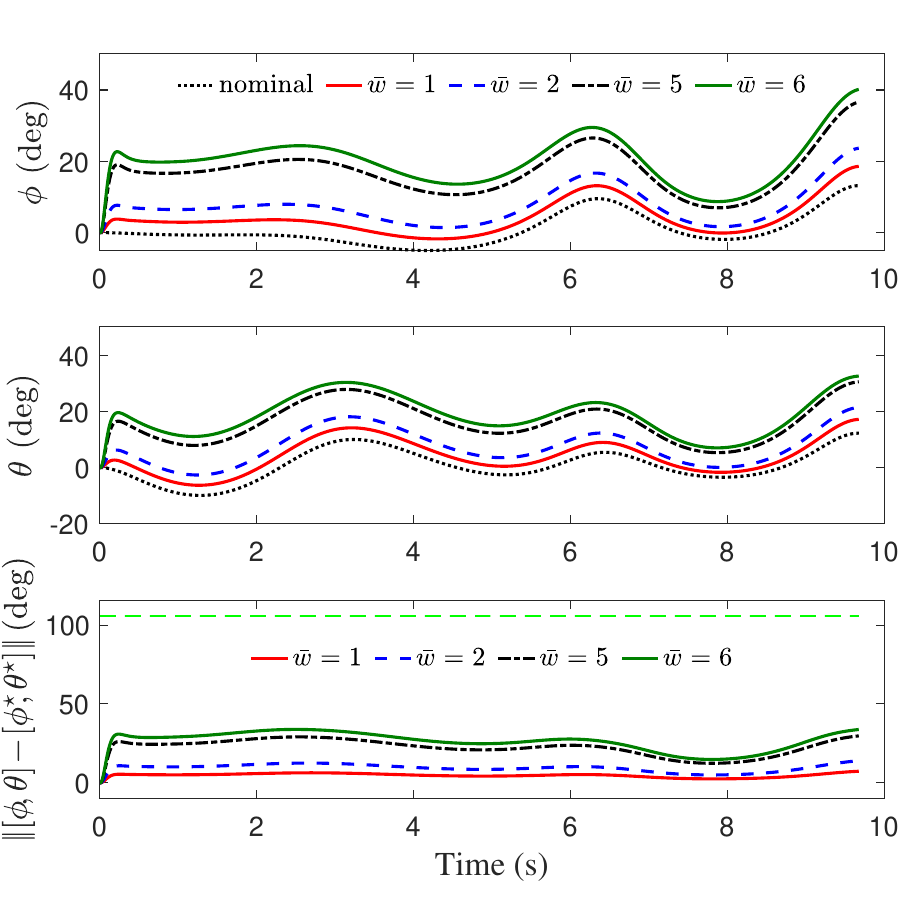}
        \vspace{-3mm}
    \caption{Nominal and actual rotational angles under RCCM-P. The green dashed lines denote the theoretical bounds associated with $\bar w= 1$.}
    \label{fig:rot-angle-quad}
        \vspace{-1mm}
\end{figure}
\begin{figure*}[h]
    \centering
        \vspace{-3mm}
     \begin{subfigure}[b]{0.4\textwidth}
		\centering
		 \includegraphics[height=6cm]{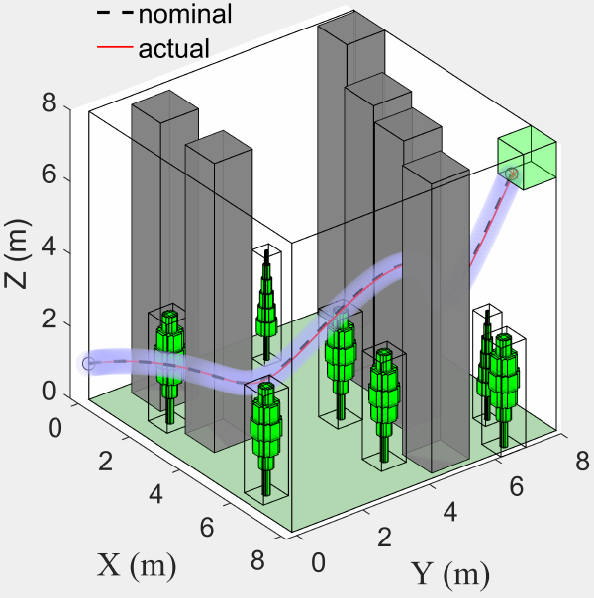}
		\caption{RCCM-P: Isometric view}
		\label{fig:traj_qaud_rccm_iso}
	\end{subfigure}%
	 \begin{subfigure}[b]{0.4\textwidth}
		\centering		     \includegraphics[height=6cm]{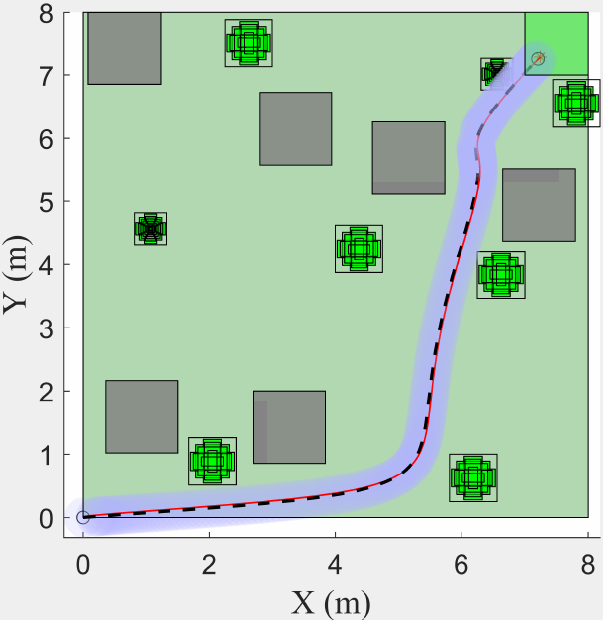}
		\caption{RCCM-P: Overhead view}
		\label{fig:traj_qaud_rccm_top}
	\end{subfigure} \vspace{3mm}\\
\begin{subfigure}[b]{0.4\textwidth}
		\centering
		 \includegraphics[height=6cm]{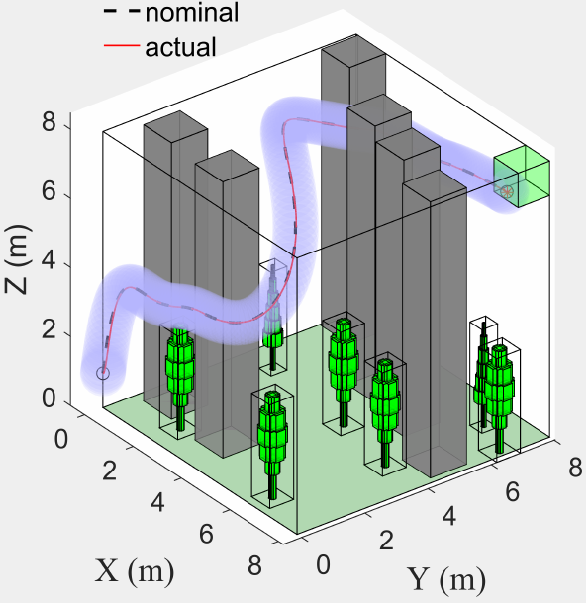}
		\caption{CCM: Isometric view}
		\label{fig:traj_qaud_ccm_iso}
	\end{subfigure}%
	 \begin{subfigure}[b]{0.4\textwidth}
		\centering	     \includegraphics[height=6cm]{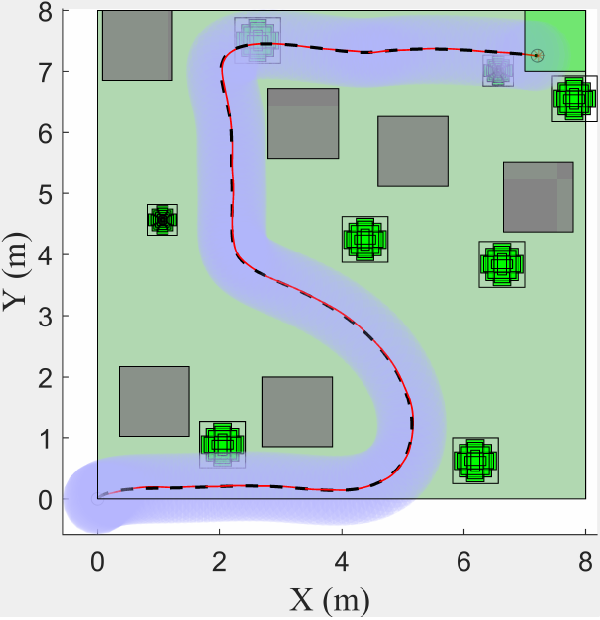}
		\caption{CCM: Overhead view}
		\label{fig:traj_qaud_ccm_top}
	\end{subfigure}
    \caption{Planned nominal and actual trajectories in an obstacle-rich environment under the RCCM-P and CCM controllers. Actual trajectories consistently stay in the (light blue shaded) ellipsoidal tubes around the nominal trajectories.}
    \label{fig:traj-qud}
\end{figure*}
\subsubsection{Feedback motion planning and tracking in cluttered environments}
\begin{figure}[h]
    \centering
        \vspace{-2mm}
    \includegraphics[width=1\columnwidth]{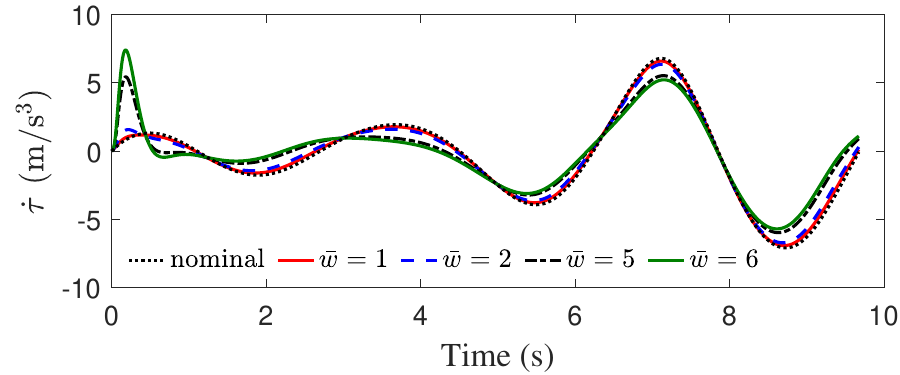}
    \caption{Nominal and actual derivatives of the total thrust (one of the control inputs)  under RCCM-P}
    \label{fig:control-quad}
    \vspace{-3mm}
\end{figure}
\begin{figure}[h]
    \centering
        \vspace{-3mm}
    \includegraphics[width=1\columnwidth]{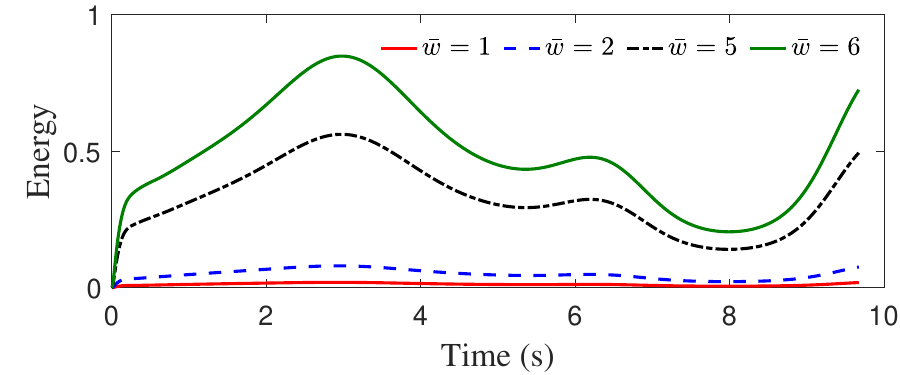}
    \caption{Riemannian energy under RCCM-P}
    \label{fig:energy}
        \vspace{-3mm}
\end{figure}
To verify the controller performance, we randomly initialized the obstacle environments for the quadrotor, depicted in \cref{fig:traj-qud}. The task for the quadrotor is to navigate from the start point $[0,0,0]^\top$ to the goal region, depicted by the light green box, while avoding collisions. %Due to space limit, we only show the performance of RCCM-P. 
We considered a wind disturbance of up to $1~\textup{m/s}^2$, i.e., $\norm{w}\leq \bar w = 1$,  under which RCCM-P yielded a tube of a 0.32 m radius ball (i.e., $\alpha_p \bar w = 0.32$) for position coordinates at $\lambda = 3.4$. For comparison, we chose the CCM controller obtained at $\hat \lambda = 3.4$, which yielded the smallest tube size gain of $0.58$ (and thus a tube of a $0.58$ m radius ball) for the positions state among all CCM controllers. 

Trajectory planning was performed by first computing
a waypoint path using geometric FMT$^\ast$ \cite{janson2015FMT-star}, and then smoothing this path using polynomial splines with the min-snap algorithm in \cite{richter2016polynomial-quadrotor}, with the Matlab codes in \cite{StanfordASL-RobustMP}. Finally, differential
flatness was leveraged to recover the open-loop (i.e., nominal) state and control trajectories. Collision checking
was performed by leveraging the configuration space representation of the obstacles, i.e., polytopes,
inflated by the 
%size of the quadrotor (approximated as a 20 cm radius ball) and the 
projection of
the tube bound onto position coordinates, i.e., $\alpha_p \norm{w}$. For simulation, CCM and RCCM-P were implemented using zero-order-hold at 250 Hz.  The Euler angular rates $(\dot \phi,\dot \theta)$ computed by the CCM/RCCM-P, and $\dot \psi$ (which was set to constant zero) were  converted to desired body rates, which were then sent to a  low-level proportional  controller. The P controller computed the three moments to track the desired body rates. The moments and the total thrust (from integrating $\dot \tau$) were then applied as ultimate inputs to the quadrotor, which consists of 12 states. With CCM and RCCM-P, the planned and actual trajectories together with the projected tubes for position coordinates under a wind disturbance, artificially simulated by $w(t)=\left(0.8+0.2\sin(0.2\pi t)\right)[\sin(45^\circ),-\cos(45^\circ),0]$,  are depicted in \cref{fig:traj-qud}. One can see that the actual position trajectories were fairly close to the nominal ones and consistently stayed in the ellipsoidal tubes under both CCM and RCCM-P. However, due to the tighter tube associated with RCCM-P, the planner was able to find a shorter collision-free trajectory with travel time of 9.68 second under RCCM-P, while the planned trajectory under CCM was more circuitous and took 13.92 second to finish. 
To evaluate the conservatism of the derived theoretical bounds associated with  RCCM-P, we performed tests under incrementally increased disturbances. The results are shown in \cref{fig:traj-err-quad,fig:rot-angle-quad,fig:control-quad,fig:energy}. According to \cref{fig:traj-err-quad}, the tracking error under a disturbance bounded by $\bar w =6$ violated the theoretical bound associated with  $\bar w=1$ for the tested disturbance setting, indicating the tube for position states was not very conservative. On the other hand, \cref{fig:rot-angle-quad} reveals that the projected tube for rotational angles, $\phi$ and $\theta$  was more conservative compared to the position states. As shown in \cref{fig:control-quad}, the trajectory of the derivative of the total thrust as one of the control inputs was quite smooth and amenable to implementation. 
% Theoretical bound (associated with the tube) is further shown  illustrate the trajectories of rotational angles, control inputs and Riemannian energy, respectively.  

\section{Conclusion}\label{sec:conclusion}
For nonlinear control-affine systems subject to bounded disturbances, this paper presents robust control contraction metrics (RCCM) for designing trajectory tracking controllers  with {explicit disturbance rejection} properties and certificate tubes for both states and inputs. The tubes are valid for any feasible nominal trajectories, and are guaranteed to contain the actual trajectories despite disturbances. Both the  RCCM  controller and the tubes can be computed, offline, by solving convex optimization problems and conveniently incorporated into a feedback motion planning framework. %We also show that our proposed RCCM approach yields tighter tubes for the states and is less conservative than an existing approach based on CCM and input-to-state stability analysis. 
Simulation results for a planar VTOL vehicle and a 3D quadrotor verify the effectiveness the proposed approach.%, which yield better tracking performance and tighter tubes compared to an existing CCM-based approach}. 

Future work includes testing of the proposed method on real hardware and leveraging the proposed method to deal with unmatched uncertainties within an adaptive control framework \cite{lakshmanan2020safe}. 
\bibliographystyle{ieeetr}
\bibliography{bib/refs}

\begin{thebibliography}{10}

\bibitem{singh2019robust}
S.~Singh, B.~Landry, A.~Majumdar, J.-J. Slotine, and M.~Pavone, ``Robust
  feedback motion planning via contraction theory,'' {\em The International
  Journal of Robotics Research, {\em submitted}}, 2019.

\bibitem{lopez2018robust-sliding}
B.~T. Lopez, J.-J. Slotine, and J.~P. How, ``Robust collision avoidance via
  sliding control,'' in {\em IEEE International Conference on Robotics and
  Automation (ICRA)}, pp.~2962--2969, 2018.

\bibitem{althoff2014online-reachability}
M.~Althoff and J.~M. Dolan, ``Online verification of automated road vehicles
  using reachability analysis,'' {\em IEEE Transactions on Robotics}, vol.~30,
  no.~4, pp.~903--918, 2014.

\bibitem{manchester2019robust-funnels}
Z.~Manchester and S.~Kuindersma, ``Robust direct trajectory optimization using
  approximate invariant funnels,'' {\em Autonomous Robots}, vol.~43, no.~2,
  pp.~375--387, 2019.

\bibitem{tedrake2010lqrtree}
R.~Tedrake, I.~R. Manchester, M.~Tobenkin, and J.~W. Roberts, ``{LQR-trees:
  Feedback motion planning via sums-of-squares verification},'' {\em The
  International Journal of Robotics Research}, vol.~29, no.~8, pp.~1038--1052,
  2010.

\bibitem{majumdar2017funnel}
A.~Majumdar and R.~Tedrake, ``Funnel libraries for real-time robust feedback
  motion planning,'' {\em The International Journal of Robotics Research},
  vol.~36, no.~8, pp.~947--982, 2017.

\bibitem{langson2004robust-tube-mpc}
W.~Langson, I.~Chryssochoos, S.~Rakovi{\'c}, and D.~Q. Mayne, ``Robust model
  predictive control using tubes,'' {\em Automatica}, vol.~40, no.~1,
  pp.~125--133, 2004.

\bibitem{mayne2005robust-tube-mpc}
D.~Q. Mayne, M.~M. Seron, and S.~Rakovi{\'c}, ``Robust model predictive control
  of constrained linear systems with bounded disturbances,'' {\em Automatica},
  vol.~41, no.~2, pp.~219--224, 2005.

\bibitem{rakovic2012parameterized-tube-mpc}
S.~V. Rakovic, B.~Kouvaritakis, M.~Cannon, C.~Panos, and R.~Findeisen,
  ``Parameterized tube model predictive control,'' {\em IEEE Transactions on
  Automatic Control}, vol.~57, no.~11, pp.~2746--2761, 2012.

\bibitem{rakovic2009set}
S.~V. Rakovi{\'c}, ``Set theoretic methods in model predictive control,'' in
  {\em Nonlinear Model Predictive Control}, pp.~41--54, Springer, 2009.

\bibitem{marruedo2002input-tube-mpc}
D.~L. Marruedo, T.~Alamo, and E.~F. Camacho, ``Input-to-state stable {MPC} for
  constrained discrete-time nonlinear systems with bounded additive
  uncertainties,'' in {\em Proc. CDC}, vol.~4, pp.~4619--4624, 2002.

\bibitem{yu2013tube-mpc}
S.~Yu, C.~Maier, H.~Chen, and F.~Allg{\"o}wer, ``Tube {MPC} scheme based on
  robust control invariant set with application to lipschitz nonlinear
  systems,'' {\em Systems \& Control Letters}, vol.~62, no.~2, pp.~194--200,
  2013.

\bibitem{lopez2019dynamic-tube-mpc}
B.~T. Lopez, J.-J.~E. Slotine, and J.~P. How, ``Dynamic tube {MPC} for
  nonlinear systems,'' in {\em Proc. ACC}, pp.~1655--1662, 2019.

\bibitem{bayer2013discrete-tube-mpc}
F.~Bayer, M.~B{\"u}rger, and F.~Allg{\"o}wer, ``{Discrete-time incremental ISS:
  A framework for robust NMPC},'' in {\em European Control Conference},
  pp.~2068--2073, IEEE, 2013.

\bibitem{kohler2020computationally-tube-mpc}
J.~K{\"o}hler, R.~Soloperto, M.~A. M{\"u}ller, and F.~Allg{\"o}wer, ``A
  computationally efficient robust model predictive control framework for
  uncertain nonlinear systems,'' {\em IEEE Transactions on Automatic Control},
  vol.~66, no.~2, pp.~794--801, 2020.

\bibitem{lohmiller1998contraction}
W.~Lohmiller and J.-J.~E. Slotine, ``On contraction analysis for non-linear
  systems,'' {\em Automatica}, vol.~34, no.~6, pp.~683--696, 1998.

\bibitem{manchester2017control}
I.~R. Manchester and J.-J.~E. Slotine, ``Control contraction metrics: {C}onvex
  and intrinsic criteria for nonlinear feedback design,'' {\em IEEE
  Transactions on Automatic Control}, vol.~62, no.~6, pp.~3046--3053, 2017.

\bibitem{tsukamoto2020neural-contraction}
H.~Tsukamoto and S.-J. Chung, ``Neural contraction metrics for robust
  estimation and control: {A} convex optimization approach,'' {\em IEEE Control
  Systems Letters}, vol.~5, no.~1, pp.~211--216, 2020.

\bibitem{tsukamoto2020robust-stochastic}
H.~Tsukamoto and S.-J. Chung, ``Robust controller design for stochastic
  nonlinear systems via convex optimization,'' {\em IEEE Transactions on
  Automatic Control}, vol.~66, no.~10, pp.~4731--4746, 2020.

\bibitem{lakshmanan2020safe}
A.~Lakshmanan, A.~Gahlawat, and N.~Hovakimyan, ``Safe feedback motion planning:
  A contraction theory and $\mathcal{L}_1$-adaptive control based approach,''
  in {\em Proceedings of 59th IEEE Conference on Decision and Control (CDC)},
  pp.~1578--1583, 2020.

\bibitem{zhao2022guaranteed-contraction-imperfect}
P.~Zhao, Z.~Guo, A.~Gahlawat, H.~Kang, and N.~Hovakimyan, ``Disturbance
  estimator-based contraction control under learned dynamics,'' {\em ArXiv
  Preprint arXiv:2112.08222}, 2022.

\bibitem{manchester2018rccm}
I.~R. Manchester and J.-J.~E. Slotine, ``Robust control contraction metrics: A
  convex approach to nonlinear state-feedback {$H_\infty$} control,'' {\em IEEE
  Control Systems Letters}, vol.~2, no.~3, pp.~333--338, 2018.

\bibitem{Scherer97Multi}
C.~Scherer, P.~Gahinet, and M.~Chilali, ``Multiobjective output-feedback
  control via {LMI} optimization,'' {\em IEEE Transactions on Automatic
  Control}, vol.~42, no.~7, pp.~896--911, 1997.

\bibitem{leung2017pseudospectral-geodesic}
K.~Leung and I.~R. Manchester, ``Nonlinear stabilization via control
  contraction metrics: {A} pseudospectral approach for computing geodesics,''
  in {\em American Control Conference}, pp.~1284--1289, 2017.

\bibitem{do2013riemannian}
M.~P. Do~Carmo, {\em Riemannian Geometry}.
\newblock Springer Science \& Business Media, 2013.

\bibitem{YALMIP}
J.~Lofberg, ``{YALMIP}: A toolbox for modeling and optimization in {MATLAB},''
  in {\em Proceedings of 2004 IEEE International Symposium on Computer Aided
  Control Systems Design}, pp.~284--289, 2004.

\bibitem{andersen2000mosek}
E.~D. Andersen and K.~D. Andersen, ``{The MOSEK interior point optimizer for
  linear programming: An implementation of the homogeneous algorithm},'' in
  {\em High Performance Optimization}, pp.~197--232, Springer, 2000.

\bibitem{kelly2017intro-traj-opt}
M.~Kelly, ``An introduction to trajectory optimization: {How} to do your own
  direct collocation,'' {\em SIAM Review}, vol.~59, no.~4, pp.~849--904, 2017.

\bibitem{Currie12opti}
J.~Currie and D.~I. Wilson, ``{OPTI: Lowering} the barrier between open source
  optimizers and the industrial {MATLAB} user,'' in {\em {Foundations of
  Computer-Aided Process Operations}} (N.~Sahinidis and J.~Pinto, eds.),
  (Savannah, Georgia, USA), 2012.

\bibitem{janson2015FMT-star}
L.~Janson, E.~Schmerling, A.~Clark, and M.~Pavone, ``Fast marching tree: {A}
  fast marching sampling-based method for optimal motion planning in many
  dimensions,'' {\em The International Journal of Robotics Research}, vol.~34,
  no.~7, pp.~883--921, 2015.

\bibitem{richter2016polynomial-quadrotor}
C.~Richter, A.~Bry, and N.~Roy, ``Polynomial trajectory planning for aggressive
  quadrotor flight in dense indoor environments,'' in {\em Robotics Research},
  pp.~649--666, Springer, 2016.

\bibitem{StanfordASL-RobustMP}
StandfordASL, ``{RobustMP}.'' \url{https://github.com/stanfordASL/RobustMP},
  2013.

\end{thebibliography}

\end{document}